\documentclass[11pt]{article}
\usepackage{amsmath,amssymb,amsfonts}
\usepackage{amsthm}
\usepackage{graphicx}
\usepackage{mathrsfs}
\usepackage{fullpage}
\usepackage{hyperref}
\usepackage{xcolor}
\usepackage{multicol}
\usepackage{url}
\usepackage[utf8]{inputenc}
\usepackage{enumitem}
\usepackage[margin=1in]{geometry}
\usepackage{cleveref}
\usepackage{doi}
\usepackage{lineno}

\usepackage{tikz}
\usepackage{tikz-qtree}
\usetikzlibrary{matrix,calc,arrows,positioning,graphs}

\usepackage[title]{appendix}
\usepackage{textcomp}
\usepackage{manyfoot}
\usepackage{booktabs}
\usepackage{algorithm}
\usepackage{algorithmicx}
\usepackage{algpseudocode}
\usepackage{listings}
\usepackage{thm-restate}
\usepackage{todonotes}
\usepackage{comment}
\usepackage{xspace}

\usepackage{tocloft}
\newtheorem{theorem}{Theorem}

\newtheorem{property}[theorem]{Property}

\newtheorem{remark}[theorem]{Remark}
\newtheorem{claim}[theorem]{Claim}
\newtheorem{lemma}[theorem]{Lemma}
\newtheorem{definition}[theorem]{Definition}

\newcommand{\lcw}{\mathbf{lcw}}
\newcommand{\ctww}{\mathbf{ctww}}
\newcommand{\rw}{\mathbf{rw}}
\newcommand{\lrw}{\mathbf{lrw}}
\newcommand{\ttww}{\mathbf{ttww}}
\newcommand{\tvtww}{\mathbf{tvtww}}
\newcommand{\lbw}{\mathbf{lbw}}

\makeatletter
\renewcommand{\paragraph}{%
  \@startsection{paragraph}{4}%
  {\z@}{6pt \@plus 1pt \@minus 1pt}{-5pt}%
  {\normalfont\normalsize\bfseries}%
}
\makeatother

\title{Improved Bounds for Twin-Width Parameter Variants with Algorithmic Applications to Counting Graph Colorings}

\author{Ambroise Baril\thanks{LORIA, Université de Lorraine, France. Contact: \href{mailto:ambroise.baril@loria.fr}{ambroise.baril@loria.fr}} 
\and Miguel Couceiro\thanks{LORIA, Université de Lorraine and INESC-ID, IST, Universidade de Lisboa. Contact: \href{mailto:miguel.couceiro@loria.fr}{miguel.couceiro@loria.fr}} 
\and Victor Lagerkvist\thanks{Linköping University, Sweden. Contact: \href{mailto:victor.lagerkvist@liu.se}{victor.lagerkvist@liu.se}}
}

\date{}

\begin{document}

\maketitle
\thispagestyle{empty}

\begin{abstract}
The {\em $H$-{\sc Coloring}} problem is a well-known generalization of the classical {\sf NP}-complete problem $k$-{\sc Coloring} where the task is to determine whether an input graph admits a homomorphism to the template graph $H$. This problem has been the subject of intense theoretical research and in
this article we study the complexity of $H$-{\sc Coloring} with respect to the parameters {\em clique-width} and the more recent {\em component twin-width}, which describe desirable computational properties of graphs. We give two surprising linear bounds between these parameters, thus improving the previously known exponential and double exponential bounds. Our constructive proof naturally extends to related parameters and as a showcase we prove that {\em total twin-width} and {\em linear clique-width} can be related via a tight quadratic bound. These bounds naturally lead to algorithmic applications. The linear bounds between component twin-width and clique-width entail natural approximations of component twin-width, by making use of the results known for clique-width. 
As for computational aspects of graph coloring, we target the richer problem of counting the number of homomorphisms to $H$ (\#$H$-{\sc Coloring}).
The first algorithm that we propose uses a contraction sequence of the input graph $G$ parameterized by the component twin-width of $G$. This leads to a positive {\sf FPT} result for the counting version. The second uses a contraction sequence of the template graph $H$ and here we instead measure the complexity with respect to the number of vertices in the input graph. Using our linear bounds we show that our algorithms are {\em always} at least as fast as the previously best \#$H$-Coloring algorithms (based on clique-width) and for several interesting classes of graphs (e.g., cographs, cycles of length $\ge 7$, or distance-hereditary graphs) are in fact strictly faster.
\end{abstract}

\maketitle

\section{Introduction} \label{sec:intro}

{\em Graph coloring} is a well-known computational problem where the goal is to color a graph in a consistent way. This problem is one of the most well-studied {\sf NP}-hard problems and enjoys a wide range of applications {\it e.g.}, in planning, scheduling, and resource allocation~\cite{formanowicz2012survey}. 
There are many variants and different formulations of the coloring problem, but the most common formulation is certainly the $k$-{\sc Coloring} problem that asks whether the vertices of an input graph can be colored using $k$ available colors in such a way that no two adjacent vertices are assigned the same color. This problem can be extended in many ways and in this paper we are particularly interested in the more general problem where any two adjacent vertices in the input graph $G$ have to be mapped to two adjacent vertices in a fixed template graph $H$ (the {\em $H$-{\sc Coloring}} problem). It is not difficult to see that $k$-{\sc Coloring} is then $K_k$-{\sc Coloring}, where $K_k$ is the $k$-vertex clique. 

The basic $H$-{\sc Coloring} problem has been extended in many directions, of which one of the most dominant formalisms is the {\em counting} extension where the task is not only to decide whether there is at least one solution (coloring) but to return the number of solutions (\#$H$-{\sc Coloring}). This framework makes it possible to encode phase transition systems modeled by partition functions, modeling problems from statistical physics such as counting $q$-particle Widom–Rowlinson configurations and counting Beach models, or the classical Ising model (for further examples, see {\it e.g.}\  Dyer \& Greenhill~\cite{dyer2000complexity}).
The $\#H$-{\sc Coloring} problem is {\sf \#P}-hard unless every connected component of $H$ is either a single vertex without a loop, a looped clique or a bipartite complete graph, and it is in {\sf P} otherwise \cite{dyer2000complexity}. The question is then to which degree we can still hope to solve it efficiently, or at least improve upon the naive bound of $|V_H|^{|V_G|}$ (where $V_H$ is the set of vertices in the template graph $H$ and $V_G$ the set of vertices in the input graph $G$).

In this article we tackle this question by targeting properties of graphs, so-called {\em graph parameters}, which give rise to efficiently solvable subproblems. We will see below several concrete examples of graph parameters, but for the moment we simply assume that each graph $G$ is associated with a number $k \in \mathbb{N}$, a {\em parameter}, which describes a structural property of $G$. Here, the idea is that  small values of $k$ correspond to graphs with a simple structure, while large values correspond to more complicated graphs. 

There are then two ways to approach intractable $H$-{\sc Coloring} problems: we either restrict the class of {\em input} graphs $G$, or the class of {\em template} graphs $H$ to graphs where the parameter is bounded by some reasonably small constant. The first task is typically studied using tools from {\em parameterized} complexity where the goal is to prove that problems are {\em fixed-parameter tractable} ({\sf FPT}), {\it i.e.}, obtaining running times of the form $f(k) \cdot \|G\|^{O(1)}$ for a computable function $f \colon \mathbb{N} \to \mathbb{N}$ (where $\|G\|$ is the number of bits required to represent the input graph $G$). The second task is more closely related to {\em fine-grained} complexity\footnote{The upper bound aspect of this field also goes under the name of ``exact exponential-time algorithms''~\cite{book2010}. Let us also remark that fine-grained complexity is also strongly associated with proving sharp lower bounds for problems in $P$.} where the goal is to prove upper and lower bounds of the form $2^{f(k)} \cdot \|G\|^{O(1)}$ for a sufficiently ``fine-grained'' parameter $k$, which in our case is always going to denote the number of vertices $|V_G|$ in the input graph $G$. Here, it is worth remarking that $H$-{\sc Coloring} is believed to be a very hard problem, and the general {\sc Coloring} problem, where the template is part of the input, is not solvable in $2^{O(|V_G|)} \cdot (\|G\|+\|H\|)^{O(1)}$ time under the {\em exponential-time hypothesis} ({\sf ETH})~\cite{fomin2015}.
Hence, regardless of whether one studies the problem under the lens of parameterized or fine-grained complexity, one needs to limit the class of considered graphs via a suitable parameter.

The most prominent graph parameter in this context is likely {\em tree-width}, which intuitively measures how close a graph is to being a tree. Bounded tree-width is in many algorithmic applications sufficient to guarantee the existence of an {\sf FPT} algorithm, but with the shortcoming of failing to capture classes of dense graphs. There are many graph parameters proposed to address this limitation of tree-width, 
and we briefly survey two noteworthy examples (see Section~\ref{sec:preliminaries} for formal definitions).

\begin{figure*}
  \begin{centering}
\center    
\scalebox{.6}{
\begin{tikzpicture}

\tikzstyle{vertex}=[draw,shape=circle];

\begin{scope}

    \node[vertex] (v1) at (90:2) {$a$};
    \node[vertex] (v2) at (-51.43+90:2) {$b$}
        edge (v1);
    \node[vertex] (v3) at (-2*51.43+90:2) {$c$}
        edge (v2);
    \node[vertex] (v4) at (-3*51.43+90:2) {$d$}
        edge (v3);
    \node[vertex] (v5) at (-4*51.43+90:2) {$e$}
        edge (v4);
    \node[vertex] (v6) at (-5*51.43+90:2) {$f$}
        edge (v5);
    \node[vertex] (v7) at (-6*51.43+90:2) {$g$}
        edge (v1)
        edge (v6);

\end{scope}

\begin{scope}[xshift=6cm]

    \node[vertex] (v1) at (90:2) {$ab$}
        edge[loop above,color=red] (v1);
    
    \node[vertex] (v3) at (-2*51.43+90:2) {$c$}
        edge[ultra thick, color=red] (v1);
    \node[vertex] (v4) at (-3*51.43+90:2) {$d$}
        edge (v3);
    \node[vertex] (v5) at (-4*51.43+90:2) {$e$}
        edge (v4);
    \node[vertex] (v6) at (-5*51.43+90:2) {$f$}
        edge (v5);
    \node[vertex] (v7) at (-6*51.43+90:2) {$g$}
        edge[ultra thick, color=red] (v1)
        edge (v6);

\end{scope}

\begin{scope}[xshift=12cm]

    \node[vertex] (v1) at (90:2) {$abc$}
        edge[loop above,color=red] (v1);
    
    \node[vertex] (v4) at (-3*51.43+90:2) {$d$}
        edge[ultra thick, color=red] (v1);
    \node[vertex] (v5) at (-4*51.43+90:2) {$e$}
        edge (v4);
    \node[vertex] (v6) at (-5*51.43+90:2) {$f$}
        edge (v5);
    \node[vertex] (v7) at (-6*51.43+90:2) {$g$}
        edge[ultra thick, color=red] (v1)
        edge (v6);

\end{scope}

\begin{scope}[xshift=-3cm,yshift=-7cm]

    \node[vertex] (v1) at (90:2) {$abcd$}
        edge[loop above,color=red] (v1);

    \node[vertex] (v5) at (-4*51.43+90:2) {$e$}
        edge[ultra thick, color=red] (v1);
    \node[vertex] (v6) at (-5*51.43+90:2) {$f$}
        edge (v5);
    \node[vertex] (v7) at (-6*51.43+90:2) {$g$}
        edge[ultra thick, color=red] (v1)
        edge (v6);

\end{scope}

\begin{scope}[xshift=3cm,yshift=-7cm]

    \node[vertex] (v1) at (90:2) {$abcde$}
        edge[loop above,color=red] (v1);

    \node[vertex] (v6) at (-5*51.43+90:2) {$f$}
        edge[ultra thick, color=red] (v1);
    \node[vertex] (v7) at (-6*51.43+90:2) {$g$}
        edge[ultra thick, color=red] (v1)
        edge (v6);

\end{scope}

\begin{scope}[xshift=9cm,yshift=-7cm]

    \node[vertex] (v1) at (90:2) {$abcdef$}
        edge[loop above,color=red] (v1);
    \node[vertex] (v7) at (-6*51.43+90:2) {$g$}
        edge[ultra thick, color=red] (v1);

\end{scope}

\begin{scope}[xshift=15cm,yshift=-7cm]

    \node[vertex] (v1) at (90:2) {$abcdefg$}
        edge[loop above,color=red] (v1);

\end{scope}

\end{tikzpicture}
}
\caption{A contraction sequence of the 7-cycle. Red edges represent an inconsistency in the merged vertex (see Section \ref{subsec:component_twin-width} for a formal definition), and the maximum red degree in the sequence thus represents the largest loss of information. }
\label{fig:Contraction sequence of C7}
\end{centering}
\end{figure*}
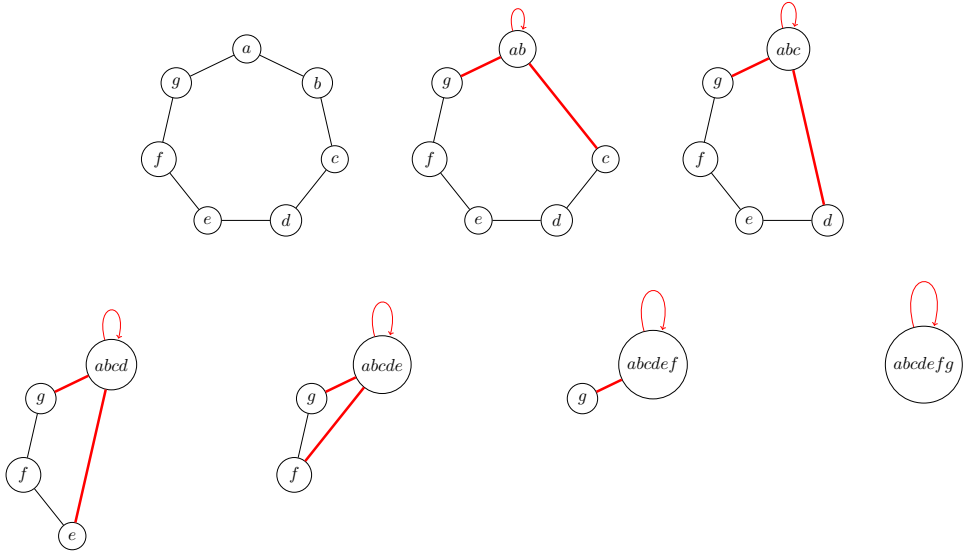

\begin{enumerate}
    \item
    {\em clique-width} ($\mathbf{cw}$). The class of graphs (with labelled vertices) with clique-width $\leq k$ is defined as the smallest class of graphs that contains the one vertex graphs $\bullet_i$ with 1 vertex labelled by $i\in [k]$, and that is stable by the following operations for $(i,j)\in [k]^2$ with $i\neq j$: $(i)$ disjoint union of graphs, $(ii)$ relabelling every vertex of label $i$ to label $j$, and $(iii)$ constructing edges between every vertex labelled by $i$ and every vertex labelled by $j$. Note that the class of cographs (which contains cliques) is exactly that of graphs with clique-width at most $2$.

    \item{\em twin-width} ($\mathbf{tww}$). The class of graphs of twin-width $\leq k$ is usually formulated via {\em contraction sequences} where graphs are gradually merged into a single vertex (see Figure~\ref{fig:Contraction sequence of C7} for an example). A graph has twin-width $\leq k$ if it admits such a contraction sequence where the maximum red degree does not exceed $k$.
  \end{enumerate}

  For clique-width, Ganian et. al~\cite{ganian2022fine} identified a structural parameter $s$ of graphs (the number of distinct non-empty intersections of neighborhoods over sets of vertices), and presented an algorithm for $H$-{\sc Coloring} that runs in $O^*(s(H)^{\mathbf{cw}(G)})$ time\footnote{The notation $O^*$ means that we ignore polynomial factors.}. It is also optimal in the sense that if there is an algorithm that solves $H$-{\sc Coloring} in time $O^*((s(H)-\varepsilon)^{\mathbf{cw}(G)})$, then the {\sf SETH} fails~\cite{ganian2022fine}. Alternative algorithms exist for templates of bounded clique-width, see Wahlstr\"om~\cite{wahlstrom2011new} who solves $\#H$-{\sc Coloring} in $O^*((2\mathbf{cw}(H)+1)^{|V_G|})$ time, and Bulatov \& Dadsetan~\cite{bulatov2020} for extensions.

Twin-width, on the other hand, is a much more recent parameter, but has in only a few years attracted significant attention~\cite{DBLP:conf/icalp/BergeBD22,bonnet2022twin8,bonnet2022twin-exponential-treewidth,bonnet2020twin3,bonnet2021twin2,bonnet2022twin7,bonnet2022twin4,bonnet2022twin6,bonnet2022twin-poly-kernels,bonnet2022reduced,bonnet2021twin-permutations,tw1,tw2,tw3,tw4,tw5,tw6,tw7,tw8,tw9,tw10,tw11}. 
One of its greatest achievements is that checking if a graph is a model of any first-order formula can be decided in {\sf FPT} time parameterized by the twin-width of the input graph. Thus, a very natural research question in light of the above results concerning tree- and clique-width is to study the complexity of $(\#)H$-{\sc Coloring} via twin-width. Unfortunately, it is easy to see that under standard assumptions, $H$-{\sc Coloring} is generally not  {\sf FPT}  parameterized by twin-width. Indeed, since twin-width is bounded on planar graphs \cite{hlinveny2022twin}, the existence of an  {\sf FPT}  algorithm for $3$-{\sc Coloring}  running in $O^*(f(\mathbf{tww}(G)))$ time implies an $O^*(1)$ time ({\it i.e.} a polynomial time) algorithm for $3$-{\sc Coloring} on planar graphs (since $f(\mathbf{tww}(G))=O(1)$ if $G$ is a planar graph). Since $3$-{\sc Coloring} is {\sf NP}-hard on planar graphs, this would imply {\sf P=NP}. Thus, $3$-{\sc Coloring} is {\em para-{\sf NP}-hard} \cite{de2017parameterized} parameterized by twin-width. 

Despite this hardness result it is possible to analyze $H$-{\sc Coloring} by a variant of twin-width known as {\em component twin-width} ($\mathbf{ctww}$) \cite{bonnet2022twin6}. This parameter equals the maximal size of a red-connected component (instead of the maximal red-degree for twin-width). It is then known that component twin-width is functionally equivalent\footnote{{\it I.e.}, each parameter is bounded by a function of the other.} to boolean-width~\cite{bonnet2022twin6}, which in turn is functionally equivalent to clique-width~\cite{bui2011boolean}.
Hence, $H$-{\sc Coloring} is  {\sf FPT}  parameterized by component twin-width, and the specific problem $k$-{\sc Coloring} is additionally known to be solvable in $O^*((2^k-1)^{\mathbf{ctww}(G)})$ time~\cite{bonnet2022twin6}. As remarked by Bonnet et al., the theoretical implications of this particular algorithm are limited due to the aforementioned (under the {\sf SETH})  optimal algorithm parameterized by clique-width~\cite{ganian2022fine}.
However, this still leaves several gaps in our understanding of component twin-width for $H$-{\sc Coloring} and its counting extension $\#H$-{\sc Coloring}.

Our paper has three major contributions to bridge these gaps. {\it Firstly,} the best known bounds between clique-width and component twin-width are obtained by following the proof of functional equivalence between component twin-width and boolean-width, and then between boolean-width and clique-width. We thereby obtain 
\begin{equation*}
\mathbf{ctww} \leq 2^{\mathbf{cw}+1}\quad \text{and} \quad \mathbf{cw}\leq 2^{2^{\mathbf{ctww}}}
\end{equation*} and $H$-{\sc Coloring} is thus solvable in $O^*(s(H)^{2^{2^{\mathbf{ctww}(G)}}})$ time. This proves  {\sf FPT}  but with a rather prohibitive run-time, and the main question is whether it is possible to improve this to a single-exponential running time $O^*(2^{O(\mathbf{ctww}(G))})$. (This line of research in parameterized complexity is relatively new but of growing importance and has seen several landmark results, see {\it e.g.}\ Chapter 11 in Cygan et al.~\cite{parameterizedbook}).
We prove that it is indeed possible by significantly strengthening the bounds between $\mathbf{cw}$ and $\mathbf{ctww}$ and obtain the linear bounds 
\[\mathbf{cw}\leq \mathbf{ctww}+1 \leq 2\mathbf{cw}.
\] Our proof is  constructive which gives a fast algorithm to derive a contraction-sequence from a clique-width expression and vice versa. To demonstrate that these ideas are not limited to these specific parameters we (in Section~\ref{sec:ttww}) consider the related problem of proving tighter bounds between {\em linear clique-width} ($\mathbf{lcw}$) and the recently introduced {\em total twin-width} ($\ttww$~\cite{bonnet2022twin6}). Linear clique-width is less explored than clique-width but comes with the advantage that faster algorithms for graph classes of bounded linear clique-width are sometimes possible (cf. the remark before Theorem 7 in Bodlaender et al.~\cite{bodlaender2023}) and that lower bounds on clique-width in many interesting cases can be generalized to linear clique-width~\cite{fomin19}. The total twin-width parameter is then known to be functionally equivalent to linear clique-width, yielding the doubly exponential bounds $ \lcw \le 2^{2^{\ttww}+1} $ and $\ttww \le (2^{\lcw}+1)(2^{\lcw-1}+1)$. We significantly improve the latter to
$$ \lcw-1 \le 2\ttww \le \lcw(\lcw+1), $$
and thus demonstrate that virtually any complexity question regarding linear clique-width can be translated to the total twin-width setting, with the possible advantage of using contraction sequences as a unifying lens. 
Specifically, it can be expected that contraction sequence related parameters are more convenient to use than (linear) clique-width, since there is only one fundamental operation to handle (vertex contraction) whereas (linear) clique-width not only deals with vertex-labelled graphs, but also introduces four fundamental operations.

{\it Secondly,} we discuss how these bounds can be exploited to {\em approximate} $\mathbf{ctww}$ by making use of the results known on $\mathbf{cw}$. Thus, an immediate consequence of our linear bounds is that $H$-{\sc Coloring} is solvable in $O^*(s(H)^{\mathbf{ctww}(G) + 1})$ time, which is a major improvement to the aforementioned double exponential upper bound.

{\it Thirdly,} we consider the generalized problem of counting the number of solutions. 
It seems unlikely that the optimal algorithm  (under SETH) by Ganian et al.~\cite{ganian2022fine} can be lifted to $\#H$-{\sc Coloring}, and while the algorithm by Wahlstr\"om~\cite{wahlstrom2011new} successfully solves $\#H$-{\sc Coloring}, it does so with the significantly worse bound of 
$2^{2\mathbf{cw}(G) \times |V_H|} (|V_G|+|V_H|)^{O(1)}$.
We tackle this problem in Section~\ref{sec:algo} by designing a novel algorithm for $\#H$-{\sc Coloring} for input graphs with bounded component twin-width and which runs in $(2^{|V_H|}-1)^{\mathbf{ctww}(G)}\times(|V_G|+|V_H|)^{O(1)}$ time. Since our linear bounds imply that $\mathbf{ctww}(H) + 2 \leq 2 \mathbf{cw}(H) + 1$ and $\mathbf{ctww}(H) + 2 \leq \mathbf{lcw}(H) + 2$ this is always at least as fast as the (linear) clique-width algorithm  by Wahlstr\"om, and strictly faster for several interesting classes of graphs. For example, cographs with edges (component twin-width 1, versus clique-width 2), cycles of length at least $7$ (component twin-width 3, versus linear clique-width 4), and distance hereditary graphs (component twin-width 3 versus clique-width 3 \cite{golumbic2000clique}).

We also consider $\#H$-{\sc Coloring} when the template graph $H$ has bounded component twin-width. We use an optimal contraction sequence of $H$ in order to obtain a $O^*((\mathbf{ctww}(H)+2)^{|V_G|})$ algorithm for $\#H$-{\sc Coloring}. For comparison, Wahlstr\"om~\cite{wahlstrom2011new} solves $\#H$-{\sc Coloring} in $O^*((2\mathbf{cw}(H)+1)^{|V_G|})$ and, slightly faster,  $O^*((\mathbf{lcw}(H)+2)^{|V_G|})$.
Due to our linear bounds we again conclude that our algorithm is always at least as fast as the $O^*((2\mathbf{cw}(H)+1)^{|V_G|})$ time clique-width algorithm by Wahlstr\"om~\cite{wahlstrom2011new}, and strictly faster for the aforementioned classes of graphs. For example, if $H$ is a cograph with edges then our algorithm solves \#$H$-{\sc coloring} in $O^*(3^{|V_G|})$ time which beats the clique-width $O^*(5^{|V_G|})$ algorithm by a significant margin.  Let us also remark that the class of cographs does not have bounded linear clique-width, so the $O^*((\mathbf{lcw}(H)+2)^{|V_G|})$ algorithm is not relevant. Also, if $H$ is a distance-hereditary graph, our algorithm solves \#$H$-{\sc coloring} in $O^*(5^{|V_G|})$ time which beats the clique-width $O^*(7^{|V_G|})$ algorithm. If $H$ is a cycle of length at least 7 we instead get $\mathbf{ctww}(H) =  3$, $\mathbf{cw}(H) = 4$, $\mathbf{lcw}(H) = 4$, yielding the bounds $O^*(5^{|V_G|})$, $O^*(9^{|V_G|})$, and $O^*(6^{|V_G|})$, {\it i.e.}, also in this case our algorithm is strictly faster.

Moreover, let us also remark that the technique employed in this article could be similarly used to derive the same results applied to the more general frameworks of counting the solutions of {\it binary constraint satisfaction problems}, {\it i.e.}, problems of the form $\#${\sc Binary-Csp}$(\Gamma)$ with a set of binary relations $\Gamma$ over a finite domain. However, to simplify the presentation  we restrict our attention to the $\#H$-{\sc Coloring} problem.

\section{Preliminaries}\label{sec:preliminaries}

Throughout this paper, a {\em graph} $G$ is a tuple $(V_G,E_G)$, where $V_G$ is a finite set (the set of vertices of $G$), and $E_G$ is a binary irreflexive symmetric relation over $V_G$ (the set of edges of $G$). A {\em looped-graph} is a $G$ is a tuple $(V_G,E_G)$, where $V_G$ is a finite set (the set of vertices of $G$), and $E_G$ is a binary symmetric relation (not necessarily irreflexive) over $V_G$ (the set of edges of $G$).
We will denote the number of vertices of a graph $G$ by $n(G)$ or, simply, by $n$ when there is no danger of ambiguity. A {\em cycle} is a graph isomorphic to the graph $C_n=( [n] , \{ (i,j)\in [n]^2 \mid |i-j|\in \{1,n-1\}\} )$ with $n\ge 3$. The neighborhood of a vertex $u$ of a graph $G$ is the set $N_G(u) = \{v\in V_G \mid (u,v)\in E_G\}$. For a graph $H$ we let {\sc $H$-Coloring} be the computational  problem of deciding whether there exists an homomorphism from an input graph $G$ to $H$, {\it i.e.}, whether there exists a function $f \colon V_G \to V_H$ such that $(x,y) \in E_G$ implies that $(f(x), f(y)) \in E_H$. We write $\#H$-{\sc Coloring} for the associated {\em counting} problem where we instead wish to determine the exact number of such homomorphisms. As remarked in Section~\ref{sec:intro}, the template graph $H$ can be chosen with great flexibility to model many different types of problems.

\subsection{Parameterized complexity}

We assume that the reader is familiar with parameterized complexity and only introduce the strictly necessary concepts (we refer to Flum \& Grohe~\cite{DBLP:series/txtcs/FlumG06} for further background). A {\em parameterized counting problem} is a pair $(F,dom)$ where $F:\Sigma^*\mapsto\mathbb{N}$ (for an alphabet $\Sigma$, {\it i.e.}, a finite set of symbols) and $dom$ is a subset of $\Sigma^*\times \mathbb{N}$. A parameterized counting problem $(F,dom)$ is said to be {\em fixed-parameter tractable} ({\sf FPT}) if there exists a computable function $f \colon \mathbb{N} \to \mathbb{N}$ such that for any instance $(x,k)\in dom$ of $F$, we can compute $F(x)$ in $f(k) \times \|x\|^{O(1)}$ time. An algorithm with this complexity is said to be an {\it {\sf FPT} algorithm}. Note that even though $f$ might be superpolynomial, which is expected if the problem is {\sf NP}-hard, instances where $k$ is reasonably small can still be efficiently solved. 

In practice, when studying  {\sf FPT}  algorithms for an {\sf NP}-hard counting problem, it is very natural to optimize the superpolynomial function $f$ that appears in the complexity of the algorithm solving it. Typically, when dealing with graph problems parameterized by the number of vertices $n$, an algorithm running in $c^n \times \|x\|^{O(1)}$ will be considered efficient in practice if $c>1$ is small. This field of research is sometimes referred to as {\em fine-grained complexity}.

\subsection{Clique-width}\label{sec:cliquewidth}

For $k\geq 1$, let $[k]=\{1,\dots, k\}$. A {\em $k$-labelled graph} $G$ is a tuple $(V_G,E_G,\ell_G)$, where $(V_G,E_G)$ is a graph and $\ell_G: V_G \to [k]$. For $i\in [k]$ and a $k$-labelled graph $G$, denote by $V_G^i= \ell_G^{-1}(\{i\})$ the set of vertices of $G$ of label $i$.
A {\em $k$-expression} $\varphi$ of a $k$-labelled graph $G$, denoted $[\varphi]=G$, is an expression defined inductively~\cite{COURCELLE199349} using:
\begin{enumerate}
\item \textbf{Single vertex:} $\bullet_i$ with $i\in [k]$: $[\bullet_i]$ is a $k$-labelled graph with one vertex labelled by $i$ (we sometimes write $\bullet_i(u)$ to state that the vertex is named $u$),
\item \textbf{Disjoint Union:} $\varphi_1\oplus \varphi_2$: $[\varphi_1\oplus\varphi_2]$ is the disjoint union of the graphs $[\varphi_1]$ and $[\varphi_2]$.
\item \textbf{Relabelling:} $\rho_{i\rightarrow j}(\varphi)$ with $(i,j)\in [k]^2$ and $i\neq j$: $[\rho_{i\rightarrow j}(\varphi)]$ is the same graph as $[\varphi]$, in which all vertices of $G$ with former label $i$ now have label $j$, 
\item \textbf{Edge Creation:} $\eta_{i,j}(\varphi)$ with $(i,j)\in [k]^2$ and $i\neq j$: $[\eta_{i,j}(\varphi)]$ is the same graph as $[\varphi]$, in which all tuples of the form $(u,v)$ with $\{\ell_G(u),\ell_G(v)\} = \{i,j\}$ is now an edge.
\end{enumerate}

A graph $G$ has a $k$-expression $\varphi$ if there exists $\ell:V_G\mapsto [k]$ such that $[\varphi]=(V_G,E_G,\ell)$. The {\em clique-width} of a graph $G$ (denoted by $\mathbf{cw}(G)$) is the minimum $k\geq 1$ such that $G$ has a $k$-expression. An {\em optimal expression} of a graph $G$ is a $\mathbf{cw}(G)$-expression of $G$.
The {\em subexpressions} of an expression $\varphi$ are defined similarly: the only subexpression of $\bullet_i$ is $\bullet_i$, the subexpressions of $\varphi=\varphi_1\oplus\varphi_2$ are $\varphi$ and the subexpressions of $\varphi_1$ and $\varphi_2$, the subexpressions of $\varphi=\rho_{i\rightarrow j}(\varphi')$ and $\varphi=\eta_{i,j}(\varphi')$ are $\varphi$ and the subexpressions of $\varphi'$.
A {\em linear} $k$-expression is a $k$-expression $\varphi$ where for every subexpression of $\varphi$ of the form $\varphi_1\oplus \varphi_2$, $\varphi_2$ is of the form $\bullet_i$ with $i\in [k]$. The {\em linear clique-width} (denoted by $\lcw(G)$) of a graph $G$ is the minimum $k\geq 1$ such that $G$ has a linear $k$-expression.

The most prominent of the many graph classes with bounded clique-width is perhaps the class of {\em cographs}: it is the class of graph that do not contain an induced path on $4$ vertices \cite{corneil1985linear}. The cographs are exactly the graphs of cliquewidth bounded by $2$ \cite{courcelle2000upper}.
Another important graph class of bounded clique-width is the class of {\em distance-hereditary graphs}: it is the class of graph in which the distances in any connected induced subgraph are the same as they are in the original graph. The class of distance-hereditary graphs strictly contains the class of cographs, and any distance-hereditary graph has its clique-width bounded by $3$ \cite{golumbic2000clique}.

\subsection{Parameters over contraction sequences}\label{subsec:component_twin-width}

Let $V$ be a finite set, and let $n:=|V|$. A {\em partition} of $V$ is a set $\mathcal{P}=\{S_1,\dots,S_k\}$ (with $k\ge 1$) of non-empty subsets of $V$, such that every element of $V$ belongs to exactly one of the $S_i$ with $i\in [k]$.
A {\em partition sequence} \cite{bonnet2022twin6} $(\mathcal{P}_n,\dots,\mathcal{P}_1)$ of $V$ is a sequence of partitions of $V$, such that $\mathcal{P}_n$ is the partition into singletons, and each $\mathcal{P}_k$ (with $k\in [n-1]$) is obtained by merging two parts of $\mathcal{P}_{k+1}$: {\it i.e.} denoting $\mathcal{P}_{k+1}=\{S_1,\dots,S_{k+1}\}$, there exists $(i,j)\in [k+1]$ with $i\neq j$ and $\mathcal{P}_k=(\mathcal{P}_{k+1}\setminus \{S_i,S_j\})\cup\{S_i\cup S_j\}$. Note that this definition implies for all $k\in [n]$, that $\mathcal{P}_k$ has $k$ elements, and that in particular, $\mathcal{P}_1=\{V\}$.

A {\em trigraph}~\cite{bonnet2020twin1} is a triplet $G=(V_G,E_G,R_G)$ where $(V_G,E_G)$ is a graph and $(V_G,R_G)$ is a looped-graph, with $E_G\cap R_G= \emptyset$. The set $E_G$ is the set of (black) edges of $G$, and $R_G$ the set of {\em red edges} of $G$. The {\em red-degree} of a vertex $u\in V_G$ is its degree in the looped-graph $(V_G,R_G)$ ignoring the red loops. A {\em red-connected component} of a trigraph $G$ is a connected component of the looped-graph $(V_G,R_G)$.
A trigraph is naturally associated to every partition of the set of vertices of a graph via the following definition.

\begin{definition}\label{def:trigraphpartition}

Let $G=(V_G,E_G)$ be a graph and $\mathcal{P}$ be a partition of $V_G$, the trigraph $G/\mathcal{P}=(\mathcal{P},E_{\mathcal{P}},R_{\mathcal{P}})$ is defined by :

\begin{itemize}
    \item $E_{\mathcal{P}}=\{ (S_1,S_2)\in\mathcal{P}^2 \mid S_1\neq S_2,  S_1\times S_2 \subseteq E_G\}$,
    \item $R_{\mathcal{P}}=(\{ (S_1,S_2)\in\mathcal{P}^2 \mid S_1\neq S_2, (S_1\times S_2)\cap E_G\neq \emptyset \}\setminus E_{\mathcal{P}})\cup \{(S,S)\mid S\in\mathcal{P}, |S|\ge 2\}$.
\end{itemize}

\end{definition}

These choices of definitions for $E_{\mathcal{P}}$ and $R_{\mathcal{P}}$ are strongly motivated by Property \ref{prop:meaning of contraction}, that enables to interpret the presence of edges between two different vertices $S_1$ and $S_2$ via the bipartite graph induced on $G$ with the bipartition $\{S_1,S_2\}$. 

\begin{property}\label{prop:meaning of contraction}
    Let $G$ be a graph, $\mathcal{P}$ be a partition of $V_G$, and let $U$ and $V$ be two different vertices of $G/\mathcal{P}$. For all $u\in U$ and $v\in V$:
  \begin{itemize}
        \item $(u,v)\in E_G$, whenever $(U,V)\in E_{G/\mathcal{P}}$, and 
        \item $(u,v)\notin E_G$, whenever $(U,V)\notin E_{G/\mathcal{P}}\cup R_{G/\mathcal{P}}$.
    \end{itemize}  
\end{property}

The presence of a black edge indicates a complete bipartite graph, whereas the absence of an edge shows that the bipartite graph has no edge. In contrast, a red edge can be viewed as a loss of complete information: it will therefore be natural to study parameters that increase with the number of red-edges. The proof of the soundness of our algorithms (in Section~\ref{sec:algo}) that make use of partition sequences rely on this fundamental property. It can be easily obtained by reformulating the definition of partition sequences.

A {\em contraction sequence} \cite{bonnet2020twin1} of a graph $G$ on at least two vertices is a sequence of trigraphs of the form $(G_n,\dots,G_1)$ with $n=|V_G|$, such that there exists a partition sequence $(\mathcal{P}_n,\dots,\mathcal{P}_1)$ with  $G_k=G/\mathcal{P}_k$, for all $k\in [n]$.
If $U$ and $V$ are the elements of $\mathcal{P}_{k+1}$ that are such that $\mathcal{P}_k=(\mathcal{P}_{k+1}\setminus \{U,V\})\cup \{U\cup V\}$, we write that $G_k=G_{k+1}/(U,V)$, as $G_k$ is obtained from $G_{k+1}$ by contracting the vertices $U$ and $V$ of $G_{k+1}$. In order to alleviate notations, we will (abusively) denote the vertex $U\cup V$ of $G_k$ as $UV$.
Note that $G_k$ has $k$ vertices and, in particular, the trigraph $G_n=(V_{G_n},E_{G_n},\emptyset)$ has no red edge, and the graph $(V_{G_n},E_{G_n})$ is isomorphic to $G$. Note also that $G_1$ has only one vertex, and is necessarily the trigraph\footnote{\label{Vertex of contraction sequence}Each vertex of $G_k$ is a set of vertices of $G$ that have been contracted.} with one vertex $G_1=(\{V_G\},\emptyset,\{(V_G,V_G)\})$. 

We can remark that a trigraph $G_k$ (with $k\in [n-1]$) obtained in a contraction sequence can be derived easily from $G_{k+1}$. The rules to follow when performing a contraction are given in Remark \ref{rem:ContractionRule}.

\begin{remark}\label{rem:ContractionRule}

For each $k\in [n-1]$, the trigraph $G_k=G_{k+1}/(U,V)$ can easily be described in function of the graph $G_{k+1}$, noticing that for all vertices $X$ and $Y$ of $G_k$:

\begin{itemize}

    \item If both $X\neq UV$ and $Y\neq UV$, $(X,Y)$ is a black edge (respectively a red edge) in $G_k$ if and only if it is a black edge (respectively red edge) in $G_{k+1}$.

    \item If $X=Y=UV$, then $(X,Y)$ is a red loop in $G_k$.

    \item If $X=UV$ and $Y\neq X$, and if both $(U,Y)$ and $(V,Y)$ are black edges in $G_{k+1}$, then $(X,Y)$ is a black edge in $G_k$.

    \item If $X=UV$ and $Y\neq X$, and if both $(U,Y)$ and $(V,Y)$ are non-edges ({\it i.e.} neither a black edge nor a red edge) in $G_{k+1}$, then $(X,Y)$ is a non-edge in $G_k$.    
    \item In any other case where $X=UV$ and $Y\neq X$, $(X,Y)$ is a red edge in $G_k$.
    
\end{itemize}

\end{remark}

To define the parameters related to contraction sequences, we introduce various notions of ``width'' for a trigraph, each of which is a function assigning an integer to any trigraph. We extend the notion of width to contraction sequences by considering the maximum width of the trigraphs occurring in the sequence. Finally, the width of a graph is defined as the minimum width among all its contraction sequences. Also, if the width notion is clear from the context, we say that a contraction sequence of a graph $G$ is {\em optimal} if its width equals the width of $G$.

The {\em twin-width} ($\mathbf{tww}$) \cite{bonnet2020twin1} of a trigraph is the maximal red-degree of its vertices. Similarly, the {\em component twin-width} ($\mathbf{ctww}$) of a trigraph is the maximal size of a red-connected component. Also, the {\em total twin-width} ($\mathbf{ttww}$)\cite{bonnet2022twin6} of trigraph is its number of red-edges.
It is known that the class of graph that admits a contraction sequence without red edges (except red loops) is exactly the class of cographs \cite{bonnet2020twin1}. As a consequence, the cographs are exactly the graphs of twin-width $0$, and of component twin-width $1$.

We also introduce a new parameter that we call the {\em total vertex twin-width}. The {\em total vertex twin-width} ($\mathbf{tvtww}$) of a trigraph is its number of vertices adjacent to at least one red edge (including red loops). 
We believe that this ``vertex-based parameter'' opens more interesting computational applications than the ``edge-based parameter'' total twin-width, as it is arguably more natural for algorithms to iterate over vertices than over edges. However, the two parameter are closely connected by natural linear and quadratic bounds.
Clearly, if a looped-graph has $t\ge 0$ vertices of degree at least $1$, it has at least $t/2$ edges and at most $t(t+1)/2$ edges. Applying these remarks to the red graphs $(V_{G'},R_{G'})$ of a trigraph $G'$ leads to the following quadratic bounds.

\begin{theorem}\label{thm:tvtww vs ttww} For any graph $G$,
$$ \tvtww(G) \le 2\ttww(G) \le (\tvtww(G))(\tvtww(G)+1). $$
\end{theorem}

\subsection{Rank-width}

A {\em branch-decomposition} \cite{oum2005graphs} of a graph $G$ is a binary tree $T$ (a tree where each non-leaf vertex has two children) whose set of leaves is exactly $V_G$. Let $G$ be a graph and $T$ a branch-decomposition of $G$. Every edge $e$ of $T$ corresponds to a bipartition $(X_e,Y_e)$ of $V_G$ by considering the bipartition of the leaves of $T$ into their connected components of $T-e$ (the tree $T$ but in which the edge $e$ have been removed). For every edge $e$ of $T$, let $A_e$ be the $\mathbb{F}_2$-matrix whose set of rows is $X_e$ and whose set of columns is $Y_e$, and whose coefficient of index $(u,v)\in X_e\times Y_e$ is $1$ if $(u,v)\in E_G$, and $0$ otherwise.

Finally, let $\rho_G(T) = \max\limits_{e\in E_T} \mathbf{rank}(A_e)$. The {\em rank-width} of $G$ denoted by $\mathbf{rw}(G)$, is the minimum of $\rho_G(T)$ for every branch-decomposition $T$ of $G$. A branch-decomposition $T$ realizing this minimum is called an {\em optimal} branch-decomposition of $G$.
We also define the {\em rank-width of a bipartition $(X,Y)$} of the vertices of $V_G$ as the rank of the $\mathbb{F}_2$-matrix $A_{X,Y}$, defined analogously to $A_e$ with respect to the bipartition $(X_e, Y_e)$.

One of the main interests of rank-width is made clear in the following lemma. 

\begin{lemma}\label{lem:similar_rows}
Let $T$ be a branch-decomposition of a graph $G$, and $e\in E_T$. If $|X_e|> 2^r$ (with $r$ the rank-width of $(X_e,Y_e)$), then there exists $(u,u')\in (X_e)^2$ with $u\neq u'$ such that 
$$N_G(u)\cap Y_e = N_G(u')\cap Y_e.$$
\end{lemma}

\begin{proof}
Since the rank of the matrix $A_e$ is $r$, the rows of $G$ all belong to a $\mathbb{F}_2$-vector space of dimension at most $r$. The latter has a cardinality of at most $2^r$, and therefore, $X_e$ has 2 identical rows, which proves the result.
\end{proof}

Also, a branch-decomposition that is a {\it caterpillar} (a rooted tree that becomes a path rooted in an extremity if the leaves are removed) is said to be a {\em linear branch-decomposition}. The {\em linear rank-width} of a graph $G$ is then the minimum of $\rho_G(T)$ for $T$ a linear branch-decomposition.

Note that giving a linear branch-decomposition of a graph is equivalent to giving a linear order over the vertices of $G$. The order $v_1\le v_2\le \dots \le v_n$ over the vertices $v_1,\dots,v_n$ of $G$ corresponds to the linear branch-decomposition given in Figure \ref{fig:LinearBranchDecomposition}.

\begin{figure}[!ht]
\begin{center}
\begin{tikzpicture}

\tikzstyle{vertex}=[circle, draw, inner sep=1pt, minimum width=6pt]
\tikzstyle{sq}=[rectangle, draw, inner sep=1pt, minimum width=6pt]

\begin{scope}

\node[vertex] (p1) at (0,0)  {};
\node[vertex] (p2) at (1,1)  {}
    edge (p1);
\node[vertex] (p3) at (2,2) {}
    edge (p2);
\node[vertex] (p4) at (3,3)  {}
    edge[dashed] (p3);

\node () at (3,3.5) {$root$};

\node[vertex] (v1) at (-1,-1) {$v_1$}
    edge (p1);
\node[vertex] (v2) at (1,-1) {$v_2$}
    edge (p1);

\node[vertex] (v3) at (2,0) {$v_3$}
    edge (p2);

\node[vertex] (v4) at (3,1) {$v_4$}
    edge (p3);

\node[vertex] (v5) at (4,2) {$v_n$}
    edge (p4);

\end{scope}

\end{tikzpicture}
\end{center}
\caption{Linear Branch decomposition corresponding to the linear order $v_1\le \dots \le v_n$.}
\label{fig:LinearBranchDecomposition}
\end{figure}
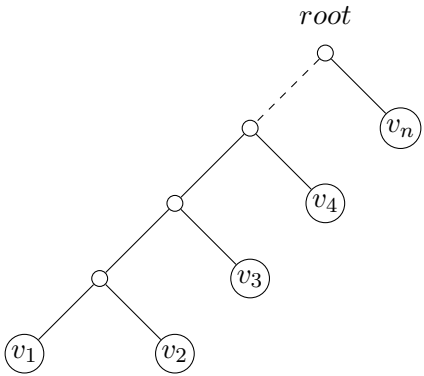

\section{Improved Bounds for Contraction Sequences Related Parameters} \label{sec:parameter_comparison}

Let us now begin the first major technical contribution of the article. In Section~\ref{sec:linear bounds} we relate  component twin-width and clique-width via a tight linear bound. As a consequence, we also manage to relate linear clique-width to component twin-width and show that the component twin-width of a graph is never higher than its linear clique-width. Then, in Section~\ref{sec:approximation} we turn to the problem of {\em approximating} component twin-width (for a given input graph). We show two positive results, one using clique-width as an intermediate parameter, and an improved approximation via rank-width. Lastly, we (in Section~\ref{sec:ttww}) prove a novel quadratic bound between total twin-width and linear clique-width. Hence, not only can  (linear) clique-width be expressed via the twin-width parameter family, but this can be accomplished with a relatively small overhead.

\subsection{Comparing clique-width and component twin-width}\label{sec:linear bounds}

In this section, we prove the linear bounds between clique-width $\mathbf{cw}$ and component twin-width $\mathbf{ctww}$. As the presence of red-loops does not impact the component twin-width, we ignore them in this section.

\begin{restatable}{theorem}{LinearBounds}
\label{thm:cw vs ctww}
For every graph $G$, $\mathbf{cw}(G) \leq \mathbf{ctww}(G)+1 \leq 2\mathbf{cw}(G)$.
\end{restatable}

Firstly, we prove the leftmost inequality. An example of the application of the proof of Lemma \ref{lem:cw leq ctww+1} is provided in Appendix \ref{app:cw vs ctww}.

\begin{restatable}{lemma}{FirstLinearBound}
\label{lem:cw leq ctww+1}

For every graph $G$, 
$\mathbf{cw}(G)\leq \mathbf{ctww}(G)+1.$

\end{restatable}

\begin{proof}
Let $(G_n,\dots,G_1)$ be an optimal contraction sequence of $G$, and let $\kappa=\mathbf{ctww}(G)$. Note that, for all $k\in [n]$, every red-connected component of $G_k$ has size $\leq \kappa$. We explain how to construct a $(\kappa+1)$-expression of $G$.

We show the following invariant for all $k\in [n]$:

\noindent$\mathbf{\mathcal{P}}(k):$ {\em ``Let $C=\{S_1,\dots,S_p\}$ be a red-connected component of $G_k$ and
$\bigcup C=S_1\cup\dots\cup S_p$. There exists a $(\kappa+1)$-expression $\varphi_C$ of the $p$-labelled graph $G_C=G[\bigcup C]$ with $\forall i\in [p], V_{G_C}^i=S_i$.''}

We first prove $\mathcal{P}(n)$. In $G_n$, there are no red edges: the red-connected components are the singletons $\{u\}$ for $u\in V_G$. Thus $\bullet_1$  is a $(\kappa+1)$-expression of $(G[\{u\}],\ell_u)$ (with $\ell_u:u\mapsto 1$), which proves $\mathcal{P}(n)$.

Now, take $k\in [n-1]$ and assume $\mathcal{P}(k+1)$. We will prove $\mathcal{P}(k)$. By definition of a contraction sequence, $G_k$ is of the form $G_k=G_{k+1}/(U,V)$ for two different vertices $U$ and $V$  of $G_{k+1}$. 

Observe that each red-connected component of $G_k$ is also a red-connected component of $G_{k+1}$, except the red-connected component $C$ containing $UV$. Hence, it suffices to prove $\mathcal{P}(k)$ for the red-connected component $C$. Notice also that $(C\setminus \{UV\}) \cup \{U,V\}$ is a union of red-connected components $C_1,\dots,C_q$ of $G_{k+1}$ (every pair of red-connected vertices in $G_{k+1}$ that does not contain $U$ or $V$ is also red-connected in $G_k$). We thus have that $C=:(C_1\cup\dots\cup C_q\cup\{UV\})\setminus \{U,V\}$.

Denote by $\{S_1,\dots,S_{p-1},S'_p\}$ the set of vertices of $C$, with $p=|C|$, and $S'_p=UV$. We have seen that $$C_1\cup\dots\cup C_q = \{S_1,\dots,S_{p-1},S_p,S_{p+1}\},$$ with $S_p:=U$ and $S_{p+1}:=V$. 

For each $i\in [p+1]$, $S_i$ belongs to a unique $C_j$ with $j\in [q]$: let $j(i)\in [q]$ be such that $S_i\in C_{j(i)}$.
By $\mathcal{P}(k+1)$ and up to interchanging  labels, for every $j\in [q]$ there exists a $(\kappa+1)$-expression $\varphi_{C_j}$ of the $p$-labelled graph $G_{C_j}=G[\bigcup C_j]$ with for all $i\in [p]$ with $j(i)=j$, $V_{G_{C_j}}^i = S_i$.
Therefore, $\varphi':=\varphi_{C_1} \oplus \dots \oplus \varphi_{C_q}$ expresses the disjoint union of the graphs $G_{C_1},...,G_{C_q}$. Furthermore, $\varphi'$ is an expression of a graph over the same vertices as $G[\bigcup C]$, Now, we still need to construct the black edges crossing these red-connected components.

We thus apply $\eta_{i,i'}$ (edge creation)\footnote{See Section \ref{sec:cliquewidth} for the notations relative to clique-width.} to $\varphi'$ for every black edge of the form $(S_i,S_{i'})$ in $G_{k+1}$, to obtain an expression $\varphi''$. Since the vertices with labels $i$ and $i'$ are exactly the vertices of $S_i$ and $S_{i'}$, we create exactly the edges between vertices of $S_i$ and of $S_{i'}$ when applying $\eta_{i,i'}$. By Property \ref{prop:meaning of contraction}, we only construct correct black edges in $G[\bigcup C]$, and thus $\varphi''$ is an expression of $G[\bigcup C]$. Conversely, as $\mathcal{P}(k+1)$ ensures that $\varphi_{C_1},\dots,\varphi_{C_q}$ represent exactly $G_{C_1},\dots,G_{C_q}$, we have that the edges of $G[\bigcup C]$ that are not represented in $\varphi'$ are exactly the edges crossing the red-connected components $C_1,\dots,C_q$ of $G_{k+1}$. In other words, the edges missing in $\varphi'$ are necessarily of the form $(a,b)\in S_{i}\times S_{i'}$, where $S_i$ and $S_{i'}$ do not belong to the same red-connected component. Since $(S_{i},S_{i'})$ is not a red edge of $G_{k+1}$ and since $(a,b)\in E_G \cap (S_{i}\times S_{i'})$, we conclude by Definition \ref{def:trigraphpartition} that $(S_{i}, S_{i'})$ is a black edge of $G_{k+1}$. Thus, $\eta_{i,i'}$ has been applied when constructing $\varphi''$, constructing thereby the edge $(a,b)$ in $\varphi''$.

Moreover, we need to make sure that the labels in $\varphi''$ match the requirements of $\mathcal{P}(k)$. For that, we set $\varphi_{G_C}:=\rho_{p+1\rightarrow p}(\varphi'')$ (relabelling). By doing so, $S_p$ (say, $U$) and $S_{p+1}$ (say, $V$) have the same label in $\varphi_{G_C}$.
Thus, it follows that $\varphi_{G_C}$ witnesses $\mathcal{P}(k)$ (since $S_p=U$ and $S_{p+1}=V$ are now contracted into $S'_p=UV$ in $G_k$) for the red-connected component $C$. Indeed, we have used $p+1=|C|+1\leq \kappa+1$ different labels to construct $\varphi_{G_C}$ from $\varphi_{C_1},\dots,\varphi_{C_q}$.
Since $\{V_G\}$ is a red-connected component of $G_1$, it follows from
$\mathcal{P}(1)$ that $G[V_G]=G$ has a $(\kappa+1)$-expression, and thus $\mathbf{cw}(G)\leq \kappa+1$. As $\kappa=\mathbf{ctww}(G)$, we have
$
\mathbf{cw}(G)\leq \mathbf{\mathbf{ctww}}(G)+1. 
$
\end{proof}

The expression of $G$ constructed in the proof of Lemma \ref{lem:cw leq ctww+1} presents an interesting structural property formalized in Claim \ref{claim:cw le ctww+1 labels}.

\begin{claim}\label{claim:cw le ctww+1 labels}

Let $\varphi_G$ be the $(\kappa+1)$-expression of the graph $G$ given by the proof of Lemma \ref{lem:cw leq ctww+1} (with $\kappa:=\ctww(G)\ge 1$). For every subexpression of $\varphi_G$ of the form $\varphi_1\oplus\varphi_2$, there are at most $\kappa$ labels that appear as the label of a vertex of $[\varphi_1]$.
    
\end{claim}

\begin{proof}

If $\varphi_1\oplus\varphi_2$ is a subexpression of $\varphi_G$, we notice that $\varphi_1$ is either of the form $\bullet_i$ with $i\in [\kappa+1]$, and $[\varphi_1]$ has therefore only $1\le \kappa$ labels, or $\varphi_1$ itself ends with a re-labelling ({\it i.e.} $\varphi_1$ is of the form $\rho_{i\rightarrow j}(\varphi_1')$). In the second case, the label $i$ can not appear as the label of a vertex of $[\varphi_1]$, which proves the claim. 
\end{proof}

Note that in contrast, $\mathbf{lcw}$ can not be bounded by a function of $\mathbf{ctww}$. For instance, the class of cographs have unbounded linear clique-width \cite{gurski2005relationship}, despite having a bounded component twin-width of $1$.
Let us now continue by proving the rightmost bound of Theorem~\ref{thm:cw vs ctww}. An example of the application of the proof of Lemma \ref{lem:ctww leq 2cw-1} is provided in Appendix \ref{app:ctww vs cw}.

\begin{restatable}{lemma}{SecondLinearBound}
\label{lem:ctww leq 2cw-1}

For every graph $G$, we have:
\begin{itemize}
\item[$(i)$] 
$\mathbf{ctww}(G) \leq 2\mathbf{cw}(G)-1$, and 
\item[$(ii)$] $\mathbf{ctww}(G)\leq \mathbf{lcw}(G).$
\end{itemize}

\end{restatable}

\begin{proof}

We first prove $(i)$\ and then adapt it to prove $(ii)$.
Let $k:=\mathbf{cw}(G)$ and take a $k$-expression of $G$. We will explain how to construct a contraction sequence of $G$ in which every red-connected component has size $\leq 2k-1$. 
The following remark will be implicitly used throughout this proof.

\begin{remark}\label{remark:same label forever}
Two vertices that have the same label in an expression $\varphi'$ also have the same label in any expression of $\varphi$ that has $\varphi'$ as a subexpression.
\end{remark}

We prove the following property of $k$-expressions of $\varphi$  by structural induction:

\noindent\emph{$\mathcal{H}(\varphi):$ ``Let $(G,\ell_G):=[\varphi]$. There exists a (partial) contraction sequence $(G_n,\dots,G_{k'})$ with $k'\leq k$ of $G$ such that:}

\begin{itemize}

    \item \emph{every red-connected component in the trigraphs $G_n,\dots,G_{k'}$ has size $\leq 2k-1$,}

    \item \emph{the vertices of $G_{k'}$ are exactly the non-empty $V_G^i$ for $i\in [k]$, and}

    \item \emph{every pair of vertices contracted have the same labels in $(G,\ell_G)$\footnote{Inductively, we say that the label of a vertex $S\in V_{G_l}$ ($k'\leq \ell \leq n$) is then the common label of the vertices that have been contracted together to produce $S$.}.''}
\end{itemize}

If $\varphi=\bullet_i$ with $i\in [k]$, there is nothing to do since $G$ has only one vertex.
If $\varphi$ is of the form $\rho_{i\rightarrow j}(\varphi')$ (with $(i,j)\in [k]^2$ and $i\neq j$), consider for $G$ the partial contraction sequence of $(G',\ell_{G'}):=[\varphi']$ given by $\mathcal{H}(\varphi')$, and then contract $V_{G'}^i$ and $V_{G'}^j$ to obtain $V_G^j = V_{G'}^i\cup V_{G'}^j$. Since $\varphi'$ is also a $k$-expression of $G$, and since that last contraction happens in a trigraph with at most $k$ vertices, this partial contraction sequence of $G$ satisfies $\mathcal{H}(\varphi)$.

If $\varphi$ is of the form $\eta_{i,j}(\varphi')$ (with $(i,j)\in [k]^2$ and $i\neq j$), consider for $G$ the partial contraction sequence of $(G',\ell_{G'}):=[\varphi']$ given by $\mathcal{H}(\varphi')$. To prove that it is sufficient to prove $\mathcal{H}(\varphi)$, it is sufficient to justify that it does not create any red edge in the contraction of $G$ that was not present in the contraction of $G'$. The first red-edge $(x,y)$ that would appear in the contraction of $G=[\eta_{i,j}(\varphi')]$ that does not appear in the same contraction of $G'=[\varphi']$, results necessarily of the contraction of two vertices $u$ and $v$ with $x=uv$ and $y$ being in the symmetric difference of the neighborhoods of $u$ and $v$ in $G=[\eta_{i,j}(\varphi')]$ but not in $G'=[\varphi']$. Such a red-edge can not exist because we contract only vertices with the same label in $\varphi'$ (or, equivalently, in $\varphi$), and that $\eta_{i,j}$ can only decrease (with respect to $\subseteq$) the symmetric difference between the neighborhood of vertices with the same label in $\varphi$. By Remark \ref{remark:same label forever}, this implies that it is also true for vertices having the same label in any subexpression of $\varphi$.

If $\varphi$ is of the form $\varphi=\varphi'\oplus \varphi''$: denote $(G',\ell'):=[\varphi']$ and $(G'',\ell''):=[\varphi'']$, thereby, $V_G=V_{G'}\cup V_{G''}$. Consider the partial contraction sequence of $G$ given by:

\begin{enumerate}

    \item contract the vertices in $V_{G'}$ in accordance to the contraction sequence given by $\mathcal{H}(\varphi')$,

    \item contract the vertices in $V_{G''}$ in accordance to the contraction sequence given by $\mathcal{H}(\varphi'')$,

    \item for all $i\in [k]$, contract  $V_{G'}^i$ with $V_{G''}^i$  (if both are nonempty) to get $V_G^i=V_{G'}^i\cup V_{G''}^i$.
    
\end{enumerate}

Steps $1$ and $2$ do not create a red-edge adjacent to both $V_{G'}$ and $V_{G''}$ (since these are two distinct connected components of $G$).
Thus, before  step $3$, we have a trigraph with $\leq 2k$ vertices (because both trigraphs obtained after $\mathcal{H}(\varphi')$ and $\mathcal{H}(\varphi'')$ have less than $k$ vertices), and every red-component that have appeared so far has size $\leq 2k-1$. After the first contraction of step $3$, the resulting trigraph has $\leq 2k-1$ vertices, and thus no red-connected component of size $>2k-1$ can emerge. Such a contraction satisfies every requirement of $\mathcal{H}(\varphi)$. We have thus proven $\mathcal{H}(\varphi)$ for every $k$-expression. 

Now, take a $k$-expression $\varphi$  of $G$. Up to applying $\rho_{i\rightarrow 1}$ for all $i\in [k]$ to $\varphi$, we can assume that $(G,\ell_G):=[\varphi]$ with $\ell_G$ being constant equal to $1$. The partial contraction sequence of $G$ given by $\mathcal{H}(\varphi)$ is a total contraction sequence of $G$ of component twin-width $\leq 2k-1$. Since $k=\mathbf{cw}(G)$, we have proven that 
$\mathbf{\mathbf{ctww}}(G)\leq 2\mathbf{cw}(G)-1.$
To prove $(ii)$, we show a similar property $\mathcal{H}_{\text{lin}}(\varphi)$ for every linear $k$-expression. The only difference between $\mathcal{H}_{\text{lin}}$ and $\mathcal{H}$ is that we replace the condition $\leq 2k-1$ (on the size of red components) by $\leq k$. The proof then follows exactly the same steps, except for the case $\varphi=\varphi'\oplus\varphi''$, where step $2$ (the contraction according to $\mathcal{H}_{\text{lin}}(\varphi'')$) is not necessary anymore, since $\varphi''$ is of the form $\bullet_i$ ($i\in [k]$), and  we obtain a trigraph of size $k+1$ instead of $2k$, since $\varphi''$ has $1$ vertex instead of $k$. This  ensures that every red-connected component has size $\leq (k+1)-1=k$ instead of $2k-1$ in the non-linear case. 

For step $3$, {\it i.e.,} contracting vertices of the same color in $\varphi'$ and in $\varphi''$, just note that it consists of at most $1$ contraction instead of $k$ in the linear case.
\end{proof}

We see that the linearity of a $k$-expression enables us to derive a stronger upper bound on the component twin-width of the graph it represents. Note that more generally, if for all subexpression of $\varphi$ of the form $\varphi_1\oplus\varphi_2$, the sum of the number of labels in $\varphi_1$ and in $\varphi_2$ does not exceed an integer $t\ge 2$, then we can conclude (with the same routine) that $\mathbf{ctww}(G)\le t-1$. This observation leads to a tight upper bound on the component twin-width of distance-hereditary graphs.

\begin{remark}\label{rem:ctwwdistancehereditary}

Let $G$ be a distance-hereditary graph. We have $\mathbf{ctww}(G)\le 3$.

\end{remark}

Indeed, if $G$ is a distance-hereditary graph, Golumbic and Rotics \cite{golumbic2000clique} witness that $\mathbf{cw}(G)\le 3$ by providing a $3$-expression $\varphi$ of  that is such that, for every subexpression of $\varphi$ of the form $\varphi_1\oplus\varphi_2$, only $2$ different labels occur in $\varphi_1$ and in $\varphi_2$.

\subsection{Approximating component twin-width} \label{sec:approximation}

The linear bounds established in Section~\ref{sec:linear bounds} entail reasonable approximation results for component twin-width by making use of known approximations of clique-width~\cite{jeong2021finding}.
The best currently known approximation algorithm for clique-width is given by Theorem \ref{thm:approx cw}.

\begin{theorem}\cite{jeong2021finding}\label{thm:approx cw}
For an input $n$-vertex graph $G$ and a positive integer $k$, we can in time
$f(k)n^3$ (for some computable function $f$) find a $(2^{k+1} - 1)$-expression of $G$ or confirm that $G$ has clique-width larger than $k$.

\end{theorem}

From Theorem \ref{thm:approx cw} and the linear bounds established in Lemma \ref{lem:cw leq ctww+1} and Lemma \ref{lem:ctww leq 2cw-1}, we immediately obtain an approximation algorithm for component twin-width.

\begin{theorem}\label{thm:Approx ctww}
For an input $n$-vertex graph $G$ and a positive integer $p$, we can in time
$f(p)n^3$ (for some computable function $f$) find a contraction sequence of $G$ of component twin-width $\leq 2^{p+3}-3$, or confirm that $G$ has component twin-width larger than $p$.

\end{theorem}

\begin{proof}
The algorithm consists of applying the algorithm of Theorem~\ref{thm:approx cw} to $G$ with $k:=p+1$. If the algorithm confirms that $\mathbf{cw}(G)> p+1$, then we know that $\mathbf{\mathbf{ctww}}(G) >p$ by Lemma \ref{lem:cw leq ctww+1}. Otherwise, it outputs a $(2^{p+2}-1)$-expression of $G$, which we transform into a contraction sequence of $G$ of component twin-width $\leq 2\times(2^{(p+2)}-1)-1 = 2^{p+3}-3$ through the constructive proof of Lemma \ref{lem:ctww leq 2cw-1}, which can be performed in linear time in the size of the $(2^{p+1}-1)$-expression of $G$.
\end{proof}

In fact, Theorem~\ref{thm:approx cw} was obtained by first comparing clique-width and rank-width (Oum and Seymour \cite{oum2006approx} proved that for any graph $G$, $\mathbf{rw}(G)\le \mathbf{cw}(G)\le 2^{\mathbf{rw}(G)+1}-1$), and, second, by using the \textsf{FPT}  algorithm (when parameterized by $k$) for the exact computation of rank-width given by the following theorem.

\begin{theorem}\cite{jeong2021finding}\label{thm:approx rw}
Given an input $n$-vertex graph $G$ and a positive integer $k$, we can find a rank-decomposition of width at most $k$ or confirm that the rank-width of $G$ is larger than $k$, in time $f(k)n^3$ (for some computable function $f$).

\end{theorem}

Thus, Theorem \ref{thm:Approx ctww} fundamentally consists in deriving bounds comparing component twin-width and rank-width from the bounds known between clique-width and rank-width, and establishes that

$$ \mathbf{rw}(G)-1 \le \mathbf{ctww}(G) \le 2^{\mathbf{rw}(G)+2} -3. $$

It is still interesting to investigate whether  a direct comparison between component twin-width and rank-width yields to better bounds, and therefore to a better approximation ratio, thanks to Theorem~\ref{thm:approx rw}. 
By avoiding using clique-width as an intermediate parameter, we can indeed prove that this is the case.

\begin{theorem}\label{thm:rw_vs_ctww}

For every graph $G$, $\mathbf{rw}(G)\le \mathbf{ctww}(G)\le 2^{\mathbf{rw}(G)+1}-1$.

\end{theorem}

We begin by first proving the leftmost bound (in Lemma \ref{lem:rw le ctww}). Note that a weaker version $\mathbf{rw}(G)-1\le \mathbf{ctww}(G)$ would follow from Lemma \ref{lem:cw leq ctww+1}, stating that $\mathbf{cw}(G)-1\le\mathbf{ctww}(G)$, and the fact that $\mathbf{rw}(G)\le\mathbf{cw}(G)$ \cite{oum2005graphs}.
To obtain that $\rw(G)\le \ctww(G)$, it is necessary to adapt the proof of the fact that $\mathbf{rw}(G)\le\mathbf{cw}(G)$ given by Oum and Seymour \cite{oum2005graphs} applied to the expression given by Lemma \ref{lem:cw leq ctww+1}, in order to take into account the structural property of the expression obtained, which is formalized in Claim \ref{claim:cw le ctww+1 labels}.

\begin{lemma}\label{lem:rw le ctww}

For every graph $G$, $\mathbf{rw}(G)\le \mathbf{ctww}(G)$.

\end{lemma}

\begin{proof}

Let $\varphi_G$ the $(\kappa+1)$-expression of $G$ given by the proof of Lemma \ref{lem:cw leq ctww+1} (with $\kappa:=\ctww(G)$), {\it i.e.} at most $\kappa+1$ different labels can appear somewhere in the definition of $\varphi_G$. 

Up to ignoring the re-labellings ($\rho_{i\rightarrow j}$ with $(i,j)\in [\kappa+1]^2$) and the edge creations $(\eta_{i,j})$, the expression $\varphi_G$ can be naturally represented by a rooted binray tree $T$, where the leaves (single vertices $\bullet_i$) are the vertices of $G$, and where the non-leaf nodes correspond to the occurences of the disjoint unions ($\oplus$). The rooted binary tree $T$ is therefore a branch-decomposition of $G$.

We show that the rank-width of $T$ is at most $\kappa$. Let $e$ be an edge of $T$. Up to interchanging $X_e$ and $Y_e$, the bipartition $(X_e,Y_e)$ is such that $X_e=V_{G_1}$, $Y_e=V_G\setminus V_{G_1}$, where $G_1=[\varphi_1]$, and where $\varphi_G$ has a subexpression of the form $\varphi_1\oplus\varphi_2$.

It is now sufficient to remark that if two vertices $u$ and $v$ of $V_{G_1}$ have the same label in $G_1$, then they have the same neighborhood (with respect to the edges in the graph $G$) in $V_G\setminus V_{G_1}$, {\it i.e.}, formally,

$$N_G(u) \cap (V_G\setminus V_{G_1}) = N_G(v) \cap (V_G\setminus V_{G_1}).$$

We have shown that two vertices of $V_{G_1}$ with the same label correspond to two identical rows in $A_e$. We have seen that, because of Claim \ref{claim:cw le ctww+1 labels}, at most $\kappa$ labels can appear as the labels of vertices of $G_1$. It follows that $A_e$ has at most $\kappa$ different rows, and therefore $\mathbf{rank}(A_e) \le \kappa$.

This is true for every edge $e$ of $T$. The branch-decomposition $T$ of $G$ witnesses that $\rw(G) \le \kappa$, with $\kappa:=\ctww(G)$.
\end{proof}

We now focus on proving the rightmost bound of Theorem \ref{thm:rw_vs_ctww} in Lemma \ref{lem:ctww le 2^rw+1-1}. The proof is very similar to one direction of the proof of functional equivalence between boolean-width and component twin-width \cite{bonnet2022twin6}, which is not surprising, since both rely exclusively on Lemma \ref{lem:similar_rows}, that applies both to rank-width and boolean-width.

\begin{lemma}\label{lem:ctww le 2^rw+1-1}

For every graph $G$, $\ctww(G) \le 2^{\rw(G)+1}-1$.

\end{lemma}

\begin{proof}

This proof follows the same scheme as the proof of the functional equivalence between boolean-width and component twin-width \cite{bonnet2022twin6}.

Similarly to a branch-decomposition of graphs, a branch-decomposition of a trigraph $G'$ is a binary tree whose set of leaves is $V_{G'}$. It is said to be rooted if a non-leaf vertex has been chosen to be the root, which leads to the usual definition of children and descendants in a rooted tree. The set of leaves descending from a vertex $v$ of a tree $T$ is denoted by $D_v^{(T)}$. Moreover, in what will follow, we will build a contraction sequence $(G_n,\dots,G_1)$ of a graph $G$, along with a sequence $(T_n,\dots,T_1)$ of branch-decomposition of $(G_n,\dots,G_1)$. We will denote $D_v^{(k)}$ instead of $D_v^{(T_k)}$ for .
Now, let $G$ be a graph and let $r:=\mathbf{rw}(G)$. We prove by downward induction (we prove $\mathcal{P}(n)$ and $\forall k\in [n-1], \mathcal{P}(k+1)\implies\mathcal{P}(k)$) the following invariant for $k\in [n]$.

$\mathcal{P}(k)$: {\em ``There exists a (partial) contraction sequence $(G_n,\dots,G_k)$ of $G$ of component twin-width $\le 2^{r+1}-1$. Moreover, there exists a branch-decomposition $T_k$ of $G_k$ such that for every $t\in V_{T_k}$ with $|D_t^{(k)}|>2^r$, there is no red-edge crossing the bipartition $(D_t^{(k)},V_{G_k}\setminus D_t^{(k)})$. Moreover, the rank-width of the bipartition $(D_t^{(k)},V_{G_k}\setminus D_t^{(k)})$ is at most $r$.''}

Note that $\mathcal{P}(n)$ is indeed true since $G=G_n$ has no red-edge, and by considering an optimal branch-decomposition of $G$.
Now assume $\mathcal{P}(k+1)$ with $k\in [n-1]$. We will prove $\mathcal{P}(k)$. First, note that if $k\le 2^r-1$, contracting any two arbitrary vertices and giving any branch-decomposition of $G_k$ proves $\mathcal{P}(k)$. We may thus assume that $k\ge 2^r$. The root $\rho$ of $T_{k+1}$ therefore satisfies $|D_{\rho}^{(k+1)}| = k+1 \ge 2^r +1$.
Observe that there exists a node $v$ of $T_{k+1}$ such that $2^r + 1 \le |D_v^{(k+1)}| \le 2^{r+1}$: a node $v$ such that $D_v^{(k+1)}$ has size at least $2^r + 1$ and which is furthest from the root meets the condition. %
By $\mathcal{P}(k+1)$, the rank-width of $(D_v^{(k+1)},V_{G_{k+1}}\setminus D_v^{(k+1)})$ is at most $r$. Using Lemma~\ref{lem:similar_rows} with respect to the edge $e$ linking $v$ to its father\footnote{If $v=\rho$ is the root, the result is trivial, as $\rho$ is then the only node with at least $2^r+1$ descendant. The root is then the only node $t$ which falls under the scope of $\mathcal{P}(k)$: the only bipartition to consider is then $(V_{G_k},\emptyset)$.} in $T_{k+1}$, there are two vertices $U$ and $U'$ of $D_v^{(k+1)}$ that satisfy $N_G(U)\cap (V_{G_{k+1}}\setminus D_v^{(k+1)}) = N_G(U')\cap (V_{G_{k+1}}\setminus D_v^{(k+1)})$. Here, the neighborhood are taken with respect to the black edges only, as by $\mathcal{P}(k+1)$ (recall that $|D_v^{(k+1)}|>2^r$ by definition of $v$), there is no red edge crossing the bipartition $(D_v^{(k+1)},V_{G_{k+1}}\setminus D_v^{(k+1)})$.

To prove $\mathcal{P}(k)$, we will prove that it is sufficient to contract the vertices $U$ and $U'$ of $G_{k+1}$ to obtain $G_k$, and to identify the leaves $U$ and $U'$ of $T_{k+1}$ to obtain $T_k$ ({\it i.e.} we remove $U'$ and shortcut every node with exactly one child that appears, and we then rename $U$ as $UU'$).
Note that all the red-edges created by the contraction of $U$ and $U'$ are adjacent to the new vertex $UU'$.

Firstly, by our choice of $U$ and $U'$, we do not create any red-edge crossing $(D_v^{(k)},V_{G_k}\setminus D_v^{(k)})$. Due to the property of $T_{k+1}$ ensured by $\mathcal{P}(k+1)$ (recall that $|D_v^{(k+1)}|>2^r$ by definition of $v$), there is no red-edge crossing $(D_v^{(k)},V_{G_k}\setminus D_v^{(k)})$ in $T_k$. The red-connected component $C$ of the new vertex $UU'$ is thus contained in $D_v^{(k)}$, and thus has size at most $|D_v^{(k)}| = |D_v^{(k+1)}|-1 \le 2^{r+1}-1$ (recall that $|D_v^{(k+1)}|\le 2^{r+1}$ by definition of $v$, and that $D_v^{(k)}$ is obtained from $D_v^{(k+1)}$ by removing $U$ and $U'$ and by adding $UU'$). Since $C$ is the only red-connected component of $G_k$ that was not a red-connected component of $G_{k+1}$, $G_k$ indeed meets the requirements of $\mathcal{P}(k)$.

Secondly, due to the choice of $v$, any node $t$ of $T_k$ with $|D_t^{(k)}|> 2^r$ containing the new node $UU'$ is an ancestor of $v$. Since $D_v^{(k)}\subseteq D_t^{(k)}$, by the above argument as for $v$, there is no red-edge crossing $(D_t^{(k)},V_{G_k}\setminus D_t^{(k)})$.

Thridly, removing a node can not make the rank-width of any bipartition of the form $(D_t^{(k)},V_{G_k}\setminus D_t^{(k)})$ with $t\in V_{G_k}$ and $|D_t^{(k)}|>2^r$ increase: it can only be lower than the rank-width of $(D_t^{(k+1)},V_{G_k}\setminus D_t^{(k+1)})$. Note also that $|D_t^{(k)}|\le |D_t^{(k+1)}|$, so if $t$ is in the scope of $\mathcal{P}(k)$, we know that it was on the scope of $\mathcal{P}(k+1)$. Therefore, by $\mathcal{P}(k+1)$, the rank-width of all such bipartitions are at most $r$.

The proof of $\mathcal{P}(k)$ is now complete: $\mathcal{P}(1)$ justifies that $\mathbf{ctww}(G)\le 2^{r+1}-1$.
\end{proof}

This bound naturally leads to the approximation given in Theorem \ref{thm:better_approx}.

\begin{theorem}\label{thm:better_approx}

For an input $n$-vertex graph $G$ and a positive integer $k$, in time
$f(k)n^3$ for some function $f$, we can find a contraction sequence of $G$ of component twin-width $\leq 2^{k+1}-1$, or confirm that $G$ has component twin-width larger than $k$.
\end{theorem}

\begin{proof}
This can be done by first applying the algorithm described in Theorem \ref{thm:approx rw}. If the algorithm outputs a branch-width $k$, we can use it to construct a contraction sequence of $G$ of component twin-width $2^{k+1}-1$ through the constructive proof of Lemma \ref{lem:ctww le 2^rw+1-1}. If the algorithm confirms that $\rw(G)\ge k$, we know by Lemma \ref{lem:rw le ctww} that $\ctww(G)\ge k$.
\end{proof}

\subsection{Comparing total twin-width and linear clique-width} \label{sec:ttww}

In this section, we provide a quadratic bound between total twin-width and linear clique-width. As discussed in Section~\ref{sec:intro} these parameters are known to be functionally equivalent, since they are both known to be functionally equivalent to {\em linear boolean-width} through the following relations \cite{oum2006approx,bonnet2022twin6}:

\begin{itemize}
    \item $ \lbw \le \lcw \le 2^{\lbw+1}  $,
    \item $ \lbw \le 2^{\ttww} $,
    \item $ \ttww \le (2^{\lbw}+1)(2^{\lbw-1}+1)$,
\end{itemize}
which entail the exponential and double-exponential bounds between linear clique-width and total twin-width:

\begin{itemize}

\item $ \ttww \le (2^{\lcw}+1)(2^{\lcw-1}+1)$,

\item $ \lcw \le 2^{2^{\ttww}+1}$.

\end{itemize}

These exponential and double exponential bounds are similar to the bounds known between component twin-width and clique-width presented in Section \ref{sec:intro}.
We improve these bounds as follows.

\begin{theorem}\label{thm:lcw vs ttww}

For every graph $G$,
$ \lcw(G)-1 \le 2\ttww(G) \le \lcw(G)(\lcw(G)+1) $.

\end{theorem}

The proof technique mirrors those of Lemma~\ref{lem:cw leq ctww+1} and Lemma \ref{lem:ctww leq 2cw-1}. Hence, our proof constructions appear to be generally applicable for showing stronger relationships between graph parameters than mere functional equivalence. We begin by first comparing linear clique-width and total vertex twin-width, and then use Theorem~\ref{thm:tvtww vs ttww}.
As we will prove, the parameter $\tvtww$ is exactly the same as $\lcw$ (up to a difference of $1$).

\begin{theorem}\label{thm:lcw vs tvtww}

For every graph $G$,
$\lcw(G)-1 \le \tvtww(G) \le \lcw(G)$.

\end{theorem}

Firstly, we show the leftmost inequality. An example of the application of the proof of Lemma \ref{lem:lcw leq tvtww+1} is provided in Appendix \ref{app:cw vs ctww}.

\begin{restatable}{lemma}{FirstBoundTotal}
\label{lem:lcw leq tvtww+1}

For every graph $G$, 
$\lcw(G)\leq \tvtww(G)+1.$

\end{restatable}

\begin{proof}
The proof is similar to the proof of Lemma~\ref{lem:cw leq ctww+1} but we include the details since the proof is constructive and has potential algorithmic applications. Let $(G_n,\dots,G_1)$ be a contraction sequence of $G$ witnessing $\kappa:=\tvtww(G)$. We explain how to construct a linear $(\kappa+1)$-expression of $G$.
We show the following invariant for all $k\in [n]$:

\noindent\emph{$\mathbf{\mathcal{P}}(k):$ ``Let $C_k=\{S_1,\dots,S_p\}$ be the set of vertices of $G_k$ of red-degree at least $1$, and
$\bigcup C_k=S_1\cup\dots\cup S_p$. There exists a linear $(\kappa+1)$-expression $\varphi_{C_k}$ of the $p$-labelled graph $G_{C_k}:=G[\bigcup C_k]$ with $V_{G_{C_k}}^i=S_i$ for all $i\in [p]$.''}

Note that for all $k\in [n]$, $|C_k|\le \kappa$ by definition of the total vertex twin-width.
We first prove $\mathcal{P}(n)$. In $G_n$, there are no red edges. Thus, $C_n=\emptyset$ and there is nothing to prove.

Now, take $k\in [n-1]$ and assume $\mathcal{P}(k+1)$. We will prove $\mathcal{P}(k)$. By definition of a contraction sequence, $G_k$ is of the form $G_k=G_{k+1}/(U,V)$ for two different vertices $U$ and $V$  of $G_{k+1}$.
First, we need to build a linear $(\kappa+1)$-expression over the right set of vertices.
Denote $C_k=\{S_1,\dots,S_{p-1},S'_p\}$ with $S'_p=UV$. Letting $S_p=U$ and $S_{p+1}=V$, we have that $S_i$ is a vertex of $G_{k+1}$ for all $i\in [p+1]$, and that 
$$\bigcup C_k=\bigcup\limits_{i=1}^{p+1}S_i.$$
Observe that $C_{k+1}$ is of the form $\{S_i\mid i\in I\}$ with $I\subseteq [p+1]$. Also, the other vertices $S_j$ with $j\in [p+1]\setminus I$ of $G_{k+1}$ are necessarily singletons. Otherwise, these vertices would have a red loop in $G_{k+1}$ (by Definition \ref{def:trigraphpartition}) and would thus belong to $C_{k+1}$. For all $j\in [p+1]\setminus I$, let $S_j = \{s_j\}$ with $s_j\in V_G$.

By $\mathcal{P}(k+1)$, up to interchanging labels, there exists a linear $(\kappa+1)$-expression $\varphi_{C_{k+1}}$ of the $|I|$-labelled graph $G_{C_{k+1}}$, such that for all $i\in I$, $V_{G_{C_{k+1}}}^i = S_i$. Therefore, $$\varphi':=\varphi_{C_{k+1}} \oplus \underset{j\in [p+1]\setminus I}{\oplus} \bullet_j(s_j)$$ is a linear expression over the same vertices of the graph $G_{C_k}$, that satisfies $V_{[\varphi']}^i = S_i$ for all $i\in [p+1]$.

Now, we still need to construct the black edges crossing the different $S_i$ for $i\in [p+1]$.
We thus apply $\eta_{i,i'}$\footnote{See Section \ref{sec:cliquewidth} for the notations relative to clique-width.} to $\varphi'$ for every black edge of the form $(S_i,S_{i'})$ in $G_{k+1}$ (with $(i,i')\in [p+1]$), to obtain an expression $\varphi''$. Since the vertices with labels $i$ and $i'$ are exactly the vertices of $S_i$ and $S_{i'}$, we create exactly the edges between vertices of $S_i$ and of $S_{i'}$ when applying $\eta_{i,i'}$ (the reasoning is similar as in the proof of Lemma \ref{lem:cw leq ctww+1}). By Property \ref{prop:meaning of contraction}, and because $\varphi_{C_{k+1}}$ is a linear expression of $G_{C_{k+1}}$, we have that $\varphi''$ is a linear expression of $G_{C_k}$.

Moreover, we need to make sure that the labels in $\varphi''$ match the requirements of $\mathcal{P}(k)$. For that, we set $\varphi_{G_{C_k}}:=\rho_{p+1\rightarrow p}(\varphi'')$. By doing so, $S_p$ (say, $U$) and $S_{p+1}$ (say, $V$) have the same label in $\varphi_{G_{C_k}}$.

Thus, it follows that $\varphi_{G_{C_k}}$ witnesses $\mathcal{P}(k)$ (since $S_p=U$ and $S_{p+1}=V$ are now contracted into $S'_p=UV$ in $G_k$). Indeed, we have used $p+1=|C_k|+1\leq \kappa+1$ different labels to construct the linear expression $\varphi_{G_{C_k}}$. The expression $\varphi_{G_{C_k}}$ is indeed linear because $\varphi_{G_{C_{k+1}}}$ is linear and because the right term of every $\oplus$ used to construct $\varphi_{C_k}$ from $\varphi_{C_{k+1}}$ is of the form $\bullet_j(s_j)$ with $s_j\in V_G$.

Since $\{V_G\}$ is a vertex of $G_1$ with a red loop (unless $G$ is a graph on $1$ vertex, in which case the theorem is trivial), it follows from
 $\mathcal{P}(1)$ that $G[V_G]=G$ has a linear $(\kappa+1)$-expression, and thus $\lcw(G)\leq \kappa+1$. As $\kappa=\tvtww(G)$, we have
$
\lcw(G)\leq \mathbf{\tvtww}(G)+1. 
$
\end{proof}

Analogously to Claim \ref{claim:cw le ctww+1 labels}, we make a structural remark on the labels of the expression built in Lemma \ref{lem:lcw leq tvtww+1}.

\begin{claim}\label{claim:lcw le tvtww+1 labels}

Let $\varphi_G$ be the linear $(\kappa+1)$-expression of the graph $G$ given by the proof of Lemma \ref{lem:cw leq ctww+1} (with $\kappa:=\tvtww(G)\ge 1$). For every subexpression of $\varphi_G$ of the form $\varphi_1\oplus\bullet_i$ with $i\in [\kappa+1]$, the label $i$ is not a label of a vertex of $[\varphi_1]$.
    
\end{claim}

We now prove the rightmost bound of Theorem \ref{thm:lcw vs tvtww}.

\begin{lemma}\label{lem:tvtww leq lcw}
For every graph $G$, we have,
$\tvtww(G) \leq \lcw(G)$

\end{lemma}

\begin{proof}

Again, we remark that the proof is similar to the proof of Lemma~\ref{lem:ctww leq 2cw-1}, but we include the details since the proof of the contraction sequence with the necessary properties is constructive and may be useful in its own right.

Let $k:=\lcw(G)$ and take a linear $k$-expression $\varphi_G$ of $G$. We will explain how to construct a contraction sequence of $G$ in which every trigraph has at most $k$ vertices of red degree at least $1$. We begin by defining the following property and then prove it by induction over $\varphi$:

\noindent$\mathcal{H}(\varphi):$ ``Let $(G,\ell_G):=[\varphi]$. There exists a (partial) contraction sequence $(G_n,\dots,G_{k'})$ of $G$ with $k'\leq k$ such that:

\begin{itemize}

    \item each of the trigraphs $G_n,\dots,G_{k'}$ have at most $k$ vertices with red degree $\ge 1$,

    \item the vertices of $G_{k'}$ are exactly the non-empty $V_G^i$ for $i\in [k]$, and

    \item every pair of vertices contracted have the same labels in $(G,\ell_G)$\footnote{Inductively, we say that the label of a vertex $S\in V_{G_l}$ ($k'\leq l\leq n$) is then the common label of the vertices that have been contracted together to produce $S$.}.''

\end{itemize}

If $\varphi=\bullet_i$ with $i\in [k]$, there is nothing to do since $G$ has only one vertex.
If $\varphi$ is of the form $\rho_{i\rightarrow j}(\varphi')$ (with $(i,j)\in [k]^2$ and $i\neq j$), consider for $G$ the partial contraction sequence of $(G',\ell_{G'}):=[\varphi']$ given by $\mathcal{H}(\varphi')$, and then contract $V_{G'}^i$ and $V_{G'}^j$ to obtain $V_G^j = V_{G'}^i\cup V_{G'}^j$. Since $\varphi'$ is also a $k$-expression of $G$, and since that last contraction happens in a trigraph with less than $k$ vertices, this partial contraction sequence of $G$ satisfies $\mathcal{H}(\varphi)$.

If $\varphi$ is of the form $\eta_{i,j}(\varphi')$ (with $(i,j)\in [k]^2$ and $i\neq j$), consider for $G$ the partial contraction sequence of $(G',\ell_{G'}):=[\varphi']$ given by $\mathcal{H}(\varphi')$. To prove that it is sufficient to prove $\mathcal{H}(\varphi)$, it is sufficient to justify that it does not create any red edge in the contraction of $G$ that was not present in the contraction of $G'$. The first red-edge $(x,y)$ that would appear in the contraction of $G=[\eta_{i,j}(\varphi')]$ that does not appear in the same contraction of $G'=[\varphi']$, results necessarily of the contraction of two vertices $u$ and $v$ with $x=uv$ and $y$ being in the symmetric difference of the neighborhoods of $u$ and $v$ in $G=[\eta_{i,j}(\varphi')]$ but not in $G'=[\varphi']$. Such a red-edge can not exist because we contract only vertices with the same label in $\varphi'$ (or, equivalently, in $\varphi$), and that $\eta_{i,j}$ can only decrease (with respect to $\subseteq$) the symmetric difference between the neighborhood of vertices with the same label in $\varphi$. By Remark \ref{remark:same label forever}, this implies that it is also true for vertices having the same label in any subexpression of $\varphi$.

If $\varphi$ is of the form $\varphi=\varphi'\oplus \bullet_i(u)$: denote $(G',\ell'):=[\varphi']$, thereby, $V_G=V_{G'}\cup \{u\}$. Consider for $G$ the partial contraction sequence obtained by performing the contractions in $(G',\ell_{G'}):=[\varphi']$ given by $\mathcal{H}(\varphi')$, and then contracting $V_{G'}^i$ (if not empty) and $u$ to obtain $V_G^i = V_{G'}^i\cup \{u\}$. Since $u$ is an isolated vertex in $G$, performing the contractions in $G'$ can not create any red edge in the contraction of $G$ that did not already exist in the contraction of $G'$. The last eventual contraction between $u$ and $V_{G'}^i$ occurs in a trigraph with at most $k+1$ vertices, resulting in a trigraph of at most $k$ vertices. In particular, there can not be more than $k$ vertices adjacent to at least one red edge. Such a contraction satisfies every requirement of $\mathcal{H}(\varphi)$. We have thus proven $\mathcal{H}(\varphi)$ for every linear $k$-expression.

Now, take a linear $k$-expression $\varphi$  of $G$. Up to applying $\rho_{i\rightarrow 1}$ for all $i\in [k]$ to $\varphi$, we can assume that $(G,\ell_G):=[\varphi]$ with $\ell_G$ being constant equal to $1$. The partial contraction sequence of $G$ given by $\mathcal{H}(\varphi)$ is a total contraction sequence of $G$ of total vertex twin-width $\le k$. Since $k=\lcw(G)$, we have thus proven that 
$$\mathbf{\tvtww}(G)\leq \lcw(G).\qedhere$$
\end{proof}

From Theorem~\ref{thm:tvtww vs ttww}  and Theorem~\ref{thm:lcw vs tvtww}  we then obtain the quadratic bound
\[\lcw -1 \le 2\ttww \le \lcw(\lcw+1)\]
of Theorem \ref{thm:lcw vs ttww}.
Moreover, as another implication of Theorem \ref{thm:lcw vs tvtww}, we can easily derive an approximation of the total vertex twin-width.
For linear clique-width we have the following approximation from Jeong et al.~\cite{jeong2017art}.

\begin{theorem}\cite{jeong2017art}\label{thm:compute lin-rw}
For an input $n$-vertex graph $G$ and a parameter $k$, we can find a linear $(2^k + 1)$-expression of $G$ confirming that $G$ has linear clique-width at most $2^k+1$ or certify that $G$ has linear clique-width larger than $k$ in time $O(f(k)n^3)$ for some computable function $f$.
\end{theorem}

From the constructive proofs of the bounds given in Theorem \ref{thm:lcw vs tvtww} we then obtain an approximation algorithm for total vertex twin-width.

\begin{theorem}

For an input $n$-vertex graph $G$ and a parameter $p$, we can find a contraction sequence of $G$ with total vertex twin-width $2^{p+1}+1$, confirming that $\tvtww(G)\le 2^{p+1}+1$ or certify that $G$ has total vertex twin-width larger than $p$ in time $O(f(p)n^3)$ for some computable function $f$.

\end{theorem}

Note that, similarly to the study we carried out in Section \ref{sec:approximation}, a direct comparison between total vertex twin-width and linear rank-width would likely result in a slightly better approximation ratio. Indeed, linear rank-width can be calculated exactly in \textsf{FPT} time.

\begin{theorem}\cite{jeong2017art}
For an input $n$-vertex graph and a parameter $k$, we can decide in time $O(f(k)n^3)$ for some function $f$ whether its linear rank-width is at most $k$ and if so, find a linear rank-decomposition of width at most $k$.
\end{theorem}

By adapting the proof of Theorem \ref{thm:rw_vs_ctww} we obtain the following bound (we omit the proof since it only involves adapting the proof of Theorem \ref{thm:rw_vs_ctww} to the new setting, and using Claim \ref{claim:lcw le tvtww+1 labels} instead of Claim \ref{claim:cw le ctww+1 labels}).

\begin{theorem}

For every graph $G$, $$\lrw(G) \le \tvtww(G) \le 2^{\lrw(G)+1}-1.$$

\end{theorem}

Finally, by combining these two results we immediately get the following approximation result for total vertex twin-width.

\begin{theorem}
For an input $n$-vertex graph $G$ and a parameter $k$, we can in $O(f(k)n^3)$ time (for some computable function $f$) witness that $\tvtww(G)\le 2^{k+1}-1$, or that $\tvtww(G)\ge k$.
\end{theorem}

\section{Complexity Results}\label{sec:algo}

In the second part of the article we show two algorithmic applications of dynamic programming over component twin-width to $\#H$-{\sc Coloring}. 
Let us remark that the proof of Lemma~\ref{lem:cw leq ctww+1} deals with component twin-width with a dynamic programming principle in the following way.

\begin{itemize}
    \item We keep track of an invariant (here, a clique-width expression) associated to every red connected component.
    \item The ``size of the invariant'' (the number of labels) grows with the number of vertices in the component.
    \item The difficulty of keeping track of the invariant though a contraction is overcome by Property \ref{prop:meaning of contraction}, that gives precise information on the structure of the edges intersecting two different red components.
\end{itemize}

We see in this section how this idea can be used to design dynamic programming algorithm in order to solve counting versions of graph coloring problems. 
The first result assumes that an optimal contraction sequence of the input graph $G$ is given, and results in a  {\sf FPT}  algorithm parameterized by $\mathbf{ctww}$, running in time $O^*((2^{|V_H|}-1)^{\mathbf{\mathbf{ctww}}(G)})$. The second approach uses an optimal contraction sequence of the template $H$ (whose computation can be seen as a pre-computation, since it does not involve the input graph $G$): we obtain a fine-grained algorithm running in time $O^*((\mathbf{\mathbf{ctww}}(H)+2)^{|V_G|})$, which outperforms the best algorithms in the literature, with a running time of $O^*((2\mathbf{cw}(H)+1)^{|V_G|})$~\cite{wahlstrom2011new}
and $O^*((\mathbf{lcw}(H)+2)^{|V_G|})$ \cite{wahlstrom2011new} 
through the linear bound of Section \ref{sec:linear bounds}. 

Note that the technique employed in this paper could similarly be used to derive the same complexity results applied to the more general frameworks of counting the solutions of {\it binary constraint satisfaction problems}, {\em i.e.} problems of the forms $\#${\sc Binary-Csp}$(\Gamma)$ with $\Gamma$ a set of binary relations over a finite domain, even though we restrict to the simpler case of $\#H$-{\sc Coloring} here to avoid having to define contraction sequences of instances and template of binary constraint satisfaction problems.

\subsection{Parameterized complexity}\label{sec:parameterized}

We present an algorithm solving $\#H$-{\sc Coloring} in  {\sf FPT}  time parameterized by component twin-width, assuming that a contraction sequence is part of the input. It is inspired by the algorithm solving $k$-{\sc Coloring} \cite{bonnet2022twin6}, thus proving that $\#H$-{\sc Coloring} is  {\sf FPT}  parameterized by component twin-width and thus also by clique-width (by functional equivalence). Throughout, we need to assume that we are given a contraction sequence of the input graph.

Let us remark that Walhstr\"om~\cite{wahlstrom2011new} solves $H$-{\sc Coloring} in time $$2^{2\mathbf{cw}(G) \times |V_H|} (|V_G|+|V_H|)^{O(1)},$$ 
whenever a $\mathbf{cw}(G)$-expression of $G$ is given. We solve it in time $$(2^{|V_H|}-1)^{\mathbf{ctww(G)}+1} \times (|V_G|+|V_H|)^{O(1)}.$$ 

However, recall that (1) $ \mathbf{ctww}(G)+1 \le  2\mathbf{cw}(G) $ by Lemma \ref{lem:ctww leq 2cw-1}, implying that our algorithm is always at least as fast, and that (2) our algorithm is strictly faster for {\it e.g.}\ cographs with edges (component twin-width 1, versus clique-width 2), cycles of length at least $7$ (component twin-width 3, versus clique-width 4), and distance-hereditary graphs that are not cographs ({\it i.e.}, those with component twin-width $\le 3$, by Remark \ref{rem:ctwwdistancehereditary}, clique-width $3$).

\begin{theorem}\label{thm:FPT algo}
For any graph $H$, there exists an algorithm running in time $$(2^{|V_H|}-1)^{\mathbf{ctww(G)}+1} \times (|V_G|+|V_H|)^{O(1)}$$  that solves $\#H$-{\sc Coloring} on any input graph $G$ (assuming that an optimal contraction sequence $(G_n,\dots,G_1)$ of $G$ is given).
\end{theorem}

\begin{proof}

For $k\in [n]$, $C=\{S_1,\dots,S_p\}\subseteq V_{G_k}$ a red-connected component of vertices of $G_k$, and for $\gamma: C\mapsto (2^{V_H}\setminus\{\emptyset\})$, an {\em $H$-coloring of $G[\cup C]$ with profile $\gamma$} is an $H$-coloring $f$ of $G[\cup C]$ such that for all $i\in [p]$, $f(S_i)=\gamma(S_i)$. {\it I.e.} the vertices of $H$ used to color $S_i$ are exactly the colors of the set $\gamma(S_i)$.

Then, define the set $COL(C,\gamma)$ as the set of $H$-colorings of $G[\cup C]$ with profile $\gamma$. We see that for every red-connected component $C$ of $G_k$, the sets $COL(C,\gamma)$ for $\gamma: C\mapsto (2^{V_H}\setminus\{\emptyset\})$ form a partition of the set of the $H$-colorings of $G[\cup C]$.

The principle of the algorithm is to inductively maintain (from $k=n$ to $1$) the knowledge of every $|COL(C,\gamma)|$ (stored in a tabular $\#col(C,\gamma)$) for each red-connected component $C$ of $G_k$ and $\gamma: C\mapsto (2^{V_H}\setminus\{\emptyset\})$. In this way, since $\{V_G\}$ is a red-connected component of $G_1$, we can obtain the number of $H$-colorings of $G[V_G]=G$ by computing $$\sum\limits_{T\in (2^{V_H}\setminus\{\emptyset\})}\#col(\{V_G\},V_G\mapsto T).$$

Firstly, note that the red-connected components of $G_n$ are the $\{u\}$ for $u\in V_G$ (since $G_n$ has no red edge). For every $\gamma:u\mapsto \gamma(u)\in (2^{|V_H|}\setminus\{\emptyset\})$ we let $\#col(\{u\},\gamma)\leftarrow 0$ if $|\gamma(u)|\neq 1$ and $\#col(\{u\},\gamma)\leftarrow 1$ if $|\gamma(u)|= 1$. Hence, we correctly store the value of $|COL(\{u\},\gamma)|$ in the tabular $\#col(\{u\},\gamma)$.

We explain how to maintain this invariant after the contraction from $G_{k+1}$ to $G_k$ (with $k\in [n-1]$). By definition of a contraction sequence, $G_k$ is of the form $G_k=: G_{k+1}/(U,V)$ with $U$ and $V$ two different vertices of $G_{k+1}$.

Note that every red-connected component of $G_k$ is also a red-connected component of $G_{k+1}$, except the red-connected component $C$ containing $UV$. We only have to compute $|COL(C,\gamma)|$ for any $\gamma: C\mapsto 2^{V_H}\setminus\{\emptyset\}$, and to store it in the tabular $\#col(C,\gamma)$. Initialize the value of $\#col(C,\gamma)$ with $0$. %

Let $C=:\{S_1\dots,S_{p-1},S'_p\}$, with $S'_p:=UV$, and $p:=|C|\leq \mathbf{\mathbf{ctww}}(G)$. Since every pair of red-connected vertices in $G_{k+1}$ (that contains neither $U$ nor $V$) are red-connected in $G_k$, $C$ must be of the form $$C:=(C_1\cup\dots\cup C_q\cup \{S'_p\})\setminus \{S_p,S_{p+1}\},$$ 
with $S_p:=U$ and $S_{p+1}:=V$ and $C_1\cup\dots\cup C_q=\{S_1,\dots,S_{p-1},S_p,S_{p+1}\}$,\footnote{Note that $UV=S'_p=S_p\cup S_{p+1}$.} and where $C_1,\dots,C_q$ (with $q>0$) are red-connected components of $G_{k+1}$ whose union contains both $S_p=U$ and $S_{p+1}=V$. Notice that each $S_i$ (for $i\in [p+1]$) belongs to a unique $C_{j(i)}$ with $j(i)\in [q]$.
These notions are illustrated in Figure \ref{fig:merging}.

The algorithm iterates over every family $(\gamma_j:C_j\mapsto (2^{V_H}\setminus\{\emptyset\}))_{1\leq j\leq q}$. Let $\gamma=\gamma_1\cup\dots\cup\gamma_q$ be the profile of $C$ that maps every $S_i$ (with $i\in [p-1]$) to $\gamma_{j(i)}(S_i)$, and that maps $S'_p=UV=S_p\cup S_{p+1}$ to $\gamma_{j(p)}(S_p)\cup\gamma_{j(p+1)}(S_{p+1})$. %
The algorithm checks if there exists a $(i,i')\in [p]^2$ with $i\neq i'$, a black edge between $S_i$ and $S_{i'}$ in $G_{k+1}$, and $(\gamma(S_i)\times\gamma(S_{i'}))\subseteq E_H$, in time $O(p^2)$. If so, we increment $\#col(C,\gamma)$ by $\prod\limits_{j=1}^q \#col(C_j,\gamma_j)$. Otherwise, we move to the next family $(\gamma_j)_{1\leq j\leq q}$.

\textbf{Soundness:} For $(\gamma_j:C_j\mapsto (2^{V_H}\setminus\{\emptyset\}))_{1\leq j\leq q}$, we denote by $COL(C,\gamma_1,\dots,\gamma_q)$ the sets of $H$-colorings $f$ of $C$ such that for all $j\in [q]$ the profile of $f|_{C_j}$ is $\gamma_j$. The algorithm is correct because, for each profile $\gamma:C\mapsto 2^{V_H}\setminus\{\emptyset\}$ of $C$, $COL(C,\gamma)$ is the disjointed union, for $(\gamma_1,\dots,\gamma_q)$ with $\gamma=\gamma_1\cup\dots\cup\gamma_q$, of the $COL(C,\gamma_1,\dots,\gamma_q)$.

We only need to compute $|COL(C,\gamma_1,\dots,\gamma_q)|$, which can be derived by Claim \ref{lemma:feasibility}. We then store the sum over $(\gamma_1,\dots,\gamma_q)$ such that $\gamma=\gamma_1\cup\dots\cup\gamma_q$ in $\#col(C,\gamma)$.
In Claim \ref{lemma:feasibility}, we say that $(\gamma_1,\dots,\gamma_q)$ is {\em feasible} if for all $(i,i')\in [p]^2$ such that $(S_i,S_{i'})$ is a black edge of $G_{k+1}, (\gamma_{j(i)}(S_i)\times\gamma_{j(i')}(S_{i'}))\subseteq E_H$.

\begin{claim}\label{lemma:feasibility}

We have for all $(\gamma_1,\dots,\gamma_q)$ that:

\begin{enumerate}
    
    \item If $(\gamma_1,\dots,\gamma_q)$ is not feasible, then $COL(C,\gamma_1,\dots,\gamma_q)=\emptyset$.
    
    \item If $(\gamma_1,\dots,\gamma_q)$ is feasible, then a function $f:\cup C\mapsto V_H$ belongs to $COL(C,\gamma_1,\dots,\gamma_q)$ if and only if, for all $j\in [q]$, $f$ restricted to $C_j$ (denoted by $f_j$) belongs to $COL(C_j,\gamma_j)$.

\end{enumerate}

\end{claim}

\begin{proof}

We treat the two cases separately.
Firstly, we assume that $(\gamma_1,\dots,\gamma_q)$ is not feasible: there exists $(i,i')\in [p]^2$ such that $(S_i,S_{i'})$ is a black edge of $G_{k+1}$ and $(\gamma_{j(i)}(S_i)\times\gamma_{j(i')}(S_{i'}))\setminus E_H\neq \emptyset$ and, for the sake of contradiction, suppose that there is $f\in COL(C,\gamma_1,\dots,\gamma_q)$. Take $(v_i,v_{i'})\in (\gamma_{j(i)}(S_i)\times\gamma_{j(i')}(S_{i'}))\setminus E_H$. By definition of a profile, there exists $(u_i,u_{i'})\in S_i\times S_{i'}$ with $f(u_i)=v_i$ and $f(u_{i'})=v_{i'}$. Then, since there exists a black edge between $S_i$ and $S_{i'}$ in $G_{k+1}$, this means by Property \ref{prop:meaning of contraction} that $(u_i,u_{i'})\in E_G$. But $(f(u_i),f(u_{i'}))=(v_i,v_{i'})\notin E_H$, so $f$ is not an $H$-coloring, which contradicts the definition of $f$.

Secondly, we assume that $(\gamma_1,\dots,\gamma_q)$ is feasible. To prove necessity, notice that the restriction of a partial $H$-coloring is also a partial $H$-coloring, and by definition of $COL(C,\gamma_1,\dots,\gamma_q)$, if $f\in COL(C,\gamma_1,\dots,\gamma_q)$, then $f_j\in COL(C_j,\gamma_j)$.

To prove sufficiency, assume that $f:\cup C\mapsto V_H$ is such that for all $j\in [q], f_j\in COL(C_j,\gamma_j)$. Then, provided that $f$ is an $H$-coloring of $G[\cup C]$, $f\in COL(C,\gamma_1,\dots,\gamma_q)$. Hence, we only have to prove that $f$ is an $H$-coloring. So let $(u,u')\in E_G$. We prove that $(f(u),f(u'))\in E_H$. Observe that there exist $S_i$ and $S_{i'}$ (with $(i,i')\in [p]^2$) such that $u\in S_i$ and $v\in S_{i'}$. If $S_i$ and $S_{i'}$ are in the same red-connected component $C_j$ (with $j\in [q]$) of $G_{k+1}$, then $(f(u),f(u'))=(f_j(u),f_j(u'))\in E_H$ because $f_j$ is an $H$-coloring. Otherwise,  $(S_i,S_{i'})$ is not a red edge of $G_{k+1}$, and since $(u,u')\in E_G$ and $(u,u')\in S_i\times S_{i'}$, it follows that so $(S_i,S_{i'})$ is a black edge of $G_{k+1}$ by Property \ref{prop:meaning of contraction}. By assumption of feasibility, $(\gamma_{j(i)}(S_i)\times\gamma_{j(i')}(S_{i'}))\subseteq E_H$ and, by definition of a profile, $(f(u),f(u')) = (f_{j(i)}(u),f_{j(i')}(u')) \in \gamma_{j(i)}(S_i)\times\gamma_{j(i')}(S_{i'}) \subseteq E_H$. The latter shows that $f$ is indeed an $H$-coloring.\end{proof}

From Claim \ref{lemma:feasibility} it follows that choosing an $f$ in $COL(C,\gamma_1,\dots,\gamma_q)$ is either impossible (if $(\gamma_1,\dots,\gamma_q)$ is not feasible), or equivalent to choosing $f_j\in COL(C_j,\gamma_j)$ for all $j\in [q]$ (in case of feasibility), which is why we add either $0$ or $\prod\limits_{j=1}^q\#col(C_j,\gamma_j)$ when treating the part of $\#col(C,\gamma)$, relatively to the feasibility of the family $(\gamma_1,\dots,\gamma_q)$.

\textbf{Complexity:} To treat the red-connected component $C$, the only non-polynomial part is to iterate over every
family $(\gamma_1,\dots,\gamma_q)$, which represents $$\prod\limits_{j=1}^q(2^{|V_H|}-1)^{|C_j|} = (2^{|V_H|}-1)^{|C|+1}\leq (2^{|V_H|}-1)^{\mathbf{\mathbf{ctww}}(G)+1}$$ families to treat (recall that for all $j\in [q]$, $\gamma_j$ is a non-empty subset of $C_j$).
\end{proof}

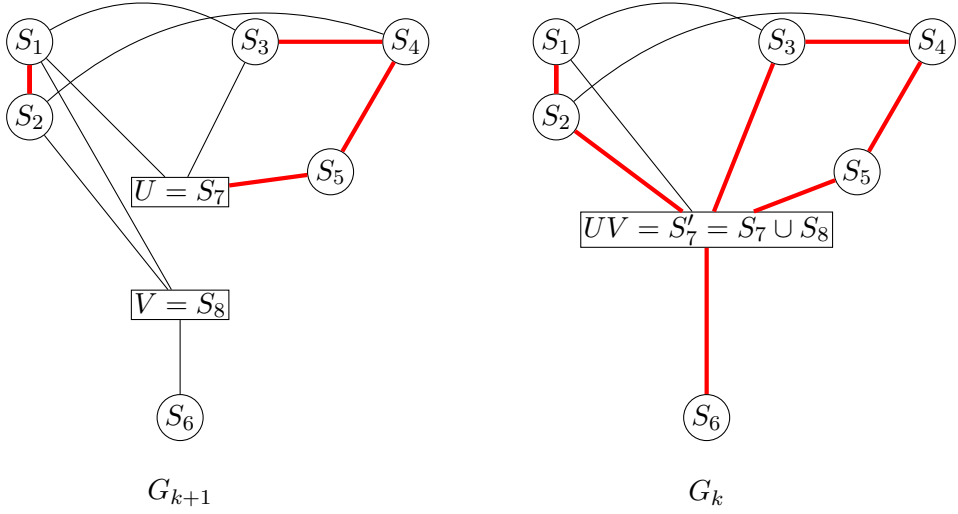
\begin{figure}[!ht]
\begin{center}
\begin{tikzpicture}

\tikzstyle{vertex}=[circle, draw, inner sep=1pt, minimum width=6pt]
\tikzstyle{sq}=[rectangle, draw, inner sep=1pt, minimum width=6pt]

\begin{scope}

\node[vertex] (w1) at (-2,3)  {$S_1$};
\node[vertex] (w2) at (-2,2)  {$S_2$};
\node[vertex] (w3) at (1,3) {$S_3$}
edge[bend right] (w1)
;
\node[vertex] (w4) at (3,3)  {$S_4$}
edge[bend right] (w2)
;
\node[vertex] (w5) at (2,1.27)  {$S_5$};

\node[vertex] (w6) at (0,-2)  {$S_6$}
;

\node[sq] (u) at (0,1) {$U=S_7$}
edge[color=red, ultra thick] (w5)
edge (w3)
edge (w1)
;

\node[sq] (v) at (0,-0.5) {$V=S_8$}
edge (w1)
edge (w2)
edge (w6)
;

\node () at (0,-3) {$G_{k+1}$}
;

\draw[color=red, ultra thick] (w1)--(w2)--(w1) (w3)--(w4)--(w5);

\end{scope}

\begin{scope}[xshift=+7cm]

\node[vertex] (w1) at (-2,3)  {$S_1$};
\node[vertex] (w2) at (-2,2)  {$S_2$};
\node[vertex] (w3) at (1,3) {$S_3$}
edge[bend right] (w1)
;
\node[vertex] (w4) at (3,3)  {$S_4$}
edge[bend right] (w2)
;
\node[vertex] (w5) at (2,1.27)  {$S_5$};

\node[vertex] (w6) at (0,-2)  {$S_6$}
;

\node[sq] (z) at (0,0.5) {$UV=S'_7=S_7\cup S_8$}
edge[color=red, ultra thick] (w5)
edge[color=red, ultra thick] (w3)
edge[color=red, ultra thick] (w6)
edge[color=red, ultra thick] (w2)
edge (w1)
;

\node () at (0,-3) {$G_k$}
;

\draw[color=red, ultra thick] (w1)--(w2)--(w1) (w3)--(w4)--(w5);

\end{scope}

\end{tikzpicture}
\end{center}
\caption{An example where contracting $U=S_7$ and $V=S_8$ causes $j=4$ different red-connected components to merge into a red-connected component of size $p=7$. With the notations of this proof, we could have  $C_1=\{S_1,S_2\}, C_2=\{S_3,S_4,S_5,S_7\}, C_3=\{S_6\}$ and $C_4=\{S_8\}$. For instance, $j(1)=j(2)=1,j(3)=j(4)=j(5)=j(7)=2,j(6)=3$ and $j(8)=4$.}
\label{fig:merging}
\end{figure}

If one only wishes to solve $H$-{\sc Coloring} rather than the counting problem, the algorithm by Ganian et al.~\cite{ganian2022fine} which runs in $O^*(s(H)^{\mathbf{cw}(G)})$ for a graph parameter $s$, is strictly more efficient. The parameter $s(H)$ counts the number of different possible non-empty common neighborhoods for a subset of vertices of $H$. Indeed, for any graph $H$, its structural parameter $s(H)$ is bounded by $2^{|V_H|}-2$ \cite{ganian2022fine} (the equality happens if and only if $H$ is a clique), and as we have proven in Lemma \ref{lem:cw leq ctww+1}, for any graph $G$, $\mathbf{cw}(G)\leq \mathbf{\mathbf{ctww}}(G)+1$. However, it appears to be difficult to extend this algorithm to the counting problem since the sets stored as invariants in the algorithm do not necessarily represent disjoint subsets of partial coloring. This is acceptable if one only wants to determine the existence of a total coloring (as long as every coloring is represented at least once), but it causes issues when counting the number of colorings.

\subsection{Fine-grained complexity}\label{sec:fine-grained}

We now consider the dual problem of solving $\#H$-{\sc Coloring} when $H$ has bounded component twin-width. We therefore use an optimal contraction sequence of the template $H$ instead of the input $G$, and obtain a fine-grained algorithm for $\#H$-{\sc Coloring} which runs in $O^*((\mathbf{\mathbf{ctww}}(H)+2)^n)$ time. 
\medskip 

\begin{theorem}\label{thm:Fine-grained algo}
$\#H$-{\sc Coloring}  is solvable in time
$O^*((\mathbf{\mathbf{ctww}}(H)+2)^{|V_G|}).$
\end{theorem}

\begin{proof}
Consider an optimal contraction sequence $(H_m,\dots,H_1)$ of $H$, with $m:=|V_H|$. Note that as $H$ is part of the template and not part of the input, an optimal contraction sequence can be precomputed (for instance by exhaustive search).
We give an algorithm similar to that described in the proof of Theorem 11, except that we define profiles for red-connected component of each $H_k$, with $k\in [m]$.

Let $C=\{T_1,\dots,T_p\}$ be a red connected component  of $H_k$ and let $\gamma=(S_1,\dots,S_p)$ be a $p$-tuple of pairwise disjoint subsets of $V_G$. An $H$-coloring $f$ of $G[S_1\cup\dots,\cup S_p]$ is said to have  {\em $C$-profile $\gamma$} if for each $i\in [p], f(S_i)\subseteq T_i$. Denote by $COL(\gamma,C)$ the set of partial $H$-colorings of $G$ ({\it i.e.,} an $H$-{\sc Coloring} of an induced subgraph) with $C$-profile $\gamma$. It is easy to compute the $|COL(\gamma,C)|$ for a red-connected component $C$ of $H_m$ (since $H_m$ has no edge) and $\gamma=(S)$ with $S\subseteq V_G$, since $C$ is of the form $C=\{v\}$ with $v\in V_H$. We have $|COL((S),\{v\})|=1$ if $S^2\cap E_G =\emptyset$, and $|COL((S),\{v\})|=0$, otherwise.

As in the proof of Theorem 11, for $k\in [m-1]$ the only red-connected component of $H_k=H_{k+1}/(U,V)$ that is not a red-connected component of $H_{k+1}$, is the red-connected component $C=\{T_1,\dots,T_{p-1},T'_p\}$ that contains $T'_p=UV$ (the vertex obtained by contraction of $T_p=U$ and $T_{p+1}=V$ in $H_{k+1}$). Hence, $C$ is of the form 
$$C=(C_1\cup\dots\cup C_q\cup \{T'_p\})\setminus \{T_p,T_{p+1}\},$$ with $C_1\cup\dots\cup C_q=\{T_1,\dots,T_{p-1},T_p,T_{p+1}\}$, where $C_1,\dots,C_q$ are the red-connected components of $H_{k+1}$ whose union contains $T_p=U$ and $T_{p+1}=V$. Again, each $T_i$ belongs to a unique $C_{j(i)}$ with $j(i)\in [q]$.

Then, as in the proof of Theorem 11, for all families of disjoint subsets of $V_G$ and $\gamma=(S_1,\dots,S_{p-1},S'_p)$, we  can compute the value of $|COL(\gamma,C)|$. Indeed, as in the proof of Theorem 11, it is the sum for every family $(\gamma_j)_{1\leq j\leq q}$ that defines the profile $\gamma$ ({\it i.e.}, each $\gamma_j$ is a family of pairwise disjoint subsets of $V_G$, and $S'_p$ is of the form $S'_p=S_p\cup S_{p+1}$ with $S_p\cap S_{p+1}=\emptyset$ and $\forall \ell\in [q]$, $\gamma_{\ell} = (S_i)_{i\in j^{-1}(\{\ell\})}$\footnote{In other words, $\gamma_{\ell}$ is the tuple of the $S_i$ where $i\in [p+1]$ is such that $T_i$ belongs to the component $C_{\ell}$.})
 of the value 
 \begin{enumerate}
 \item $\prod\limits_{j=1}^q |COL(\gamma_j,C_j)|$ if $(\gamma_1,\dots,\gamma_q)$ is feasible,
 \item $0$ otherwise.
 \end{enumerate}
Here we say that $(\gamma_1,\dots,\gamma_q)$ is {\em feasible} if for every $(i,i')\in [p]^2$ with $j(i)\neq j(i')$ and for every edge $(u_i,u_{i'})$ of $G$ with $u_i\in S_i$ and $u_{i'}\in S_{i'}$, there is a black edge between $T_i$ and $T_{i'}$ in $H_{k+1}$,

The complexity of computing $|COL(\gamma,C)|$ for every $\gamma$ is  $(\mathbf{\mathbf{ctww}}(H)+2)^{|V_G|}$, since exploring every family $(\gamma_j)_{1\leq j\leq q}$ containing only pairwise disjoint subsets of $|V_G|$ requires to explore $(\sum\limits_{j=1}^q |C_j|+1)^{|V_G|}$ families (any vertex of $G$ can be mapped to a unique element in $\{T_1,T_2,\ldots,T_{p+1}\}$ or none of them), which makes $(|C|+2)^n\leq (\mathbf{\mathbf{ctww}}(H)+2)^n$ possibilities.
Since $V_H$ is a red connected component of $H_1$,
we obtain the number of such $H$-colorings of $G$ in time $O^*((\mathbf{\mathbf{ctww}}(H)+2)^{|V_G|})$, and it is equal to $|COL(\{V_G\},\{V_H\})|$.
\end{proof}

We again remark that, by Lemma \ref{lem:ctww leq 2cw-1}, $$\mathbf{\mathbf{ctww}}(H)+2 \leq \mathbf{lcw}(H)+2$$ and 
$\mathbf{\mathbf{ctww}}(H)+2\leq 2\mathbf{cw}(H)+1$ for any graph $H$. Therefore, the algorithm in the proof of Theorem \ref{thm:Fine-grained algo} is always at least as fast  as the clique-width approach by Wahlstr\"om \cite{wahlstrom2011new}, and as remarked in Section~\ref{sec:intro}, it is strictly faster for e.g.\ cographs with edges and cycles of length $\ge 7$, and distance hereditary graphs that are not cographs by Remark \ref{rem:ctwwdistancehereditary}.

\section{Conclusion and Future Research}
\label{section:conclusions}

In this article we explored component twin-width in the context of $\#H$-{\sc Coloring} problems.
We improved the bounds of the functional equivalence between component twin-width and clique-width from the (doubly) exponential bound 
$$\mathbf{cw} \leq 2^{2^{\mathbf{ctww}}} \quad \text{and} \quad \mathbf{ctww}\leq 2^{\mathbf{cw}+1}$$ to the linear bounds $$\mathbf{cw}\leq \mathbf{ctww}+1 \leq 2\mathbf{cw}.$$ 
In particular, this entails a single-exponential {\sf FPT} algorithm for $H$-{\sc Coloring} parameterized by component twin-width. 
From these linear bounds derives an approximation algorithm with exponential ratio, that can even be improved by a direct comparison with rank-width. We then demonstrated that our constructive proof technique could be extended to related parameters, and proved a quadratic bound between total twin-width and linear clique-width.

Finally, we turned to algorithmic applications, and constructed two algorithms for solving $\#H$-{\sc Coloring}. The first uses a given optimal contraction sequence of the input graph $G$ to solve $\#H$-{\sc Coloring} in {\sf FPT} time parameterized by component twin-width. The second uses a contraction sequence of the template graph $H$ and beats the clique-width approach for solving $\#H$-{\sc Coloring} (with respect to $|V_G|$).
Let us now discuss some topics for future research.

{\bf Tightness of the bounds.}
Even though the bound $\mathbf{cw}\le\mathbf{ctww}+1$ given by Lemma \ref{lem:cw leq ctww+1} is tight for any cograph with at least 1 edge, we do not currently know if this bound can be improved for graphs with greater clique-width or component twin-width. Moreover, it would be interesting to determine whether the bound $\mathbf{ctww}\le 2\mathbf{cw}-1$ given by Lemma \ref{lem:ctww leq 2cw-1} is tight. In particular, we believe that identifying classes of graphs, such as distance-hereditary graphs, for which a similar reasoning to the one presented in Remark \ref{rem:ctwwdistancehereditary} applies, constitutes a promising direction for future research. The same remark on tightness holds for the bounds between component twin-width and rank-width given by Theorem \ref{thm:rw_vs_ctww}. It would be interesting to study the tightness of the bound $\mathbf{tww} \le 2\mathbf{cw}-2$ (where $\mathbf{tww}$ designs the twin-width), which is a direct consequence of Lemma \ref{lem:ctww leq 2cw-1}.
Also, since Lemmas \ref{lem:lcw leq tvtww+1} and \ref{lem:tvtww leq lcw}
provide very tight bounds, it is natural to ask for the characterization of the classes of graphs where each bound is attained.

{\bf Lower bounds on complexity.}
The algorithms relying on clique-width to solve $H$-{\sc Coloring} by \cite{ganian2022fine} in  $O^*(s(H)^{\mathbf{cw}(G)})$ time are known to be optimal under the SETH. We have a similar optimality result for tree-width ($\mathbf{tw}$), with an algorithm solving $H$-{\sc Coloring} in time $|V_H|^{\mathbf{tw}(G)}$, and the existence of a $(|V_H|-\varepsilon)^{\mathbf{tw}(G)}$ algorithm with $\varepsilon >0$ being ruled out under SETH. A natural research direction is then to optimize the running time of the algorithm of Theorem \ref{thm:FPT algo}, possibly by making use of $s(H)$, and prove a similar lower bound.

{\bf Extensions.}
Instead of solving $\#H$-{\sc Coloring} the results of Section~\ref{sec:algo} can be extended to arbitrary binary constraints ({\em binary constraint satisfaction problems}, {\sc Bcsp}s). The notion of component twin-width indeed generalizes naturally to both instances and templates of a {\sc Bcsp}. A natural continuation is then to investigate {\em infinite-domain {\sc Bcsp}s} which are frequently used to model problems of interest in qualitative temporal and spatial reasoning. Here, there are only a handful of results using the much weaker tree-width parameter~\cite{DBLP:conf/aaai/DabrowskiJOO21}, so an {\sf FPT} algorithm using component twin-width or clique-width would be a great generalization.
Additionally, one may note that the algorithms detailed in Section~\ref{sec:algo} can be adapted to solve a ``cost'' version of $\#H$-{\sc Coloring}: given a weight matrix $C$, the cost of a homomorphism $f$ is $\sum\limits_{u\in V_G} C(u,f(u))$, and we want to find a homomorphism of minimal cost. Can this be extended to other types of generalized problems?

\bibliographystyle{abbrv}
\bibliography{biblio}

\appendix

\section{Converting Contraction Sequences to $k$-expression and vice-versa}\label{app:ContractionSequences}

In this appendix we provide a visual example of the constructive proof of the bounds of Theorem \ref{thm:cw vs ctww}.

\LinearBounds*

We first illustrate the lefmost inequality in Section \ref{app:cw vs ctww}, and then illustrate the rightmost part in Section \ref{app:ctww vs cw}

\subsection{From contraction sequences to $k$-expressions}\label{app:cw vs ctww}

We begin by recalling Lemma~\ref{lem:cw leq ctww+1}.

\FirstLinearBound*

As input the method takes a contraction sequence of the same graph that witnesses that its component twin-width is $\le \kappa$, and and uses it do describe a $(\kappa+1)$-expression of a graph.
Also note that the  same construction illustrates Lemma \ref{lem:lcw leq tvtww+1} over the same graph (however, red-loops will not be represented).

\FirstBoundTotal*

\begin{figure}[H]
\centering

\scalebox{1.2}{

\begin{tikzpicture}
\tikzstyle{place}=[draw,shape=circle];

\begin{scope}

    \node[place] (a) at (0,0) {$a$};
    \node[place] (b) at (0,1) {$b$};
    \node[place] (c) at (0,2) {$c$};
    \node[place] (d) at (1,0) {$d$};
    \node[place] (e) at (1,1) {$e$};
    \node[place] (f) at (1,2) {$f$};
    \node[place] (g) at (2,1) {$g$};

    \draw (a)--(b)--(c)--(f)--(g)--(d)--(a)--(f)--(b)--(e)--(g)--(d)--(e)--(c)--(b)--(d);
    
\end{scope}

\begin{scope}[xshift=5cm]

    \node[place,fill=blue] (a) at (0,0) {$a$};
    \node[place,fill=blue] (b) at (0,1) {$b$};
    \node[place,fill=blue] (c) at (0,2) {$c$};
    \node[place,fill=blue] (d) at (1,0) {$d$};
    \node[place,fill=blue] (e) at (1,1) {$e$};
    \node[place,fill=blue] (f) at (1,2) {$f$};
    \node[place,fill=blue] (g) at (2,1) {$g$};

\draw (-.4,-.4)--(-.4,.4)--(.4,.4)--(.4,-.4)--(-.4,-.4);%

\draw (-.4,1-.4)--(-.4,1.4)--(.4,1.4)--(.4,1-.4)--(-.4,1-.4);%

\draw (-.4,2-.4)--(-.4,2.4)--(.4,2.4)--(.4,2-.4)--(-.4,2-.4);%

\draw (1-.4,-.4)--(1-.4,.4)--(1.4,.4)--(1.4,-.4)--(1-.4,-.4);%

\draw (1-.4,1-.4)--(1-.4,1.4)--(1.4,1.4)--(1.4,1-.4)--(1-.4,1-.4);%

\draw (1-.4,2-.4)--(1-.4,2.4)--(1.4,2.4)--(1.4,2-.4)--(1-.4,2-.4);%

\draw (2-.4,1-.4)--(2-.4,1.4)--(2.4,1.4)--(2.4,1-.4)--(2-.4,1-.4);%

\end{scope}

\begin{scope}[xshift=10cm,yshift=-.5cm]

\node () at (0,3) {$\varphi_{\textcolor{blue}{a}} = \textcolor{blue}{\bullet}$};

\node () at (0,2.5) {$\varphi_{\textcolor{blue}{b}} = \textcolor{blue}{\bullet}$};

\node () at (0,2) {$\varphi_{\textcolor{blue}{c}} = \textcolor{blue}{\bullet}$};

\node () at (0,1.5) {$\varphi_{\textcolor{blue}{d}} = \textcolor{blue}{\bullet}$};

\node () at (0,1) {$\varphi_{\textcolor{blue}{e}} = \textcolor{blue}{\bullet}$};

\node () at (0,.5) {$\varphi_{\textcolor{blue}{f}} = \textcolor{blue}{\bullet}$};

\node () at (0,0) {$\varphi_{\textcolor{blue}{g}} = \textcolor{blue}{\bullet}$};

\end{scope}

\end{tikzpicture}

}

\caption{Initial situation. All vertices are blue, but we can interchange labels within an expression if necessary.}
\end{figure}
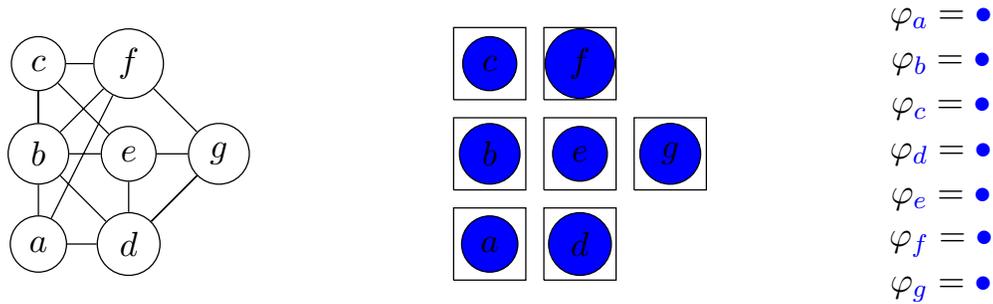

\begin{figure}[H]
\centering

\scalebox{1.2}{

\begin{tikzpicture}
\tikzstyle{place}=[draw,shape=circle];
\tikzstyle{square}=[draw,shape=rectangle];

\begin{scope}

    \node[place] (a) at (0,0) {$a$};
    \node[place] (b) at (0,1) {$b$};
    \node[place] (c) at (0,2) {$c$};
    \node[place] (d) at (1,0) {$d$};
    \node[place] (e) at (1,1) {$e$};
    \node[place] (f) at (1,2) {$f$};
    \node[place] (g) at (2,1) {$g$};

    \draw (a)--(b)--(c)--(f)--(g)--(d)--(a)--(f)--(b)--(e)--(g)--(d)--(e)--(c)--(b)--(d);
    
\end{scope}

\begin{scope}[xshift=5cm]

    \node[place,fill=blue] (a) at (0,0) {$a$};
    \node[place,fill=blue] (b) at (0,1) {$b$};
    \node[place,fill=blue] (c) at (0,2) {$c$};
    \node[place,fill=green] (d) at (1,0) {$d$};
    \node[place,fill=pink] (e) at (1,1) {$e$};
    \node[place,fill=yellow] (f) at (1,2) {$f$};
    \node[place,fill=blue] (g) at (2,1) {$g$};

\draw (-.4,-.4)--(-.4,.4)--(.4,.4)--(.4,-.4)--(-.4,-.4);%

\draw (-.4,1-.4)--(-.4,1.4)--(.4,1.4)--(.4,1-.4)--(-.4,1-.4);%

\draw (-.4,2-.4)--(-.4,2.4)--(.4,2.4)--(.4,2-.4)--(-.4,2-.4);%

\draw (1-.4,-.4)--(1-.4,.4)--(1.4,.4)--(1.4,-.4)--(1-.4,-.4);%

\draw (1-.4,1-.4)--(1-.4,1.4)--(1.4,1.4)--(1.4,1-.4)--(1-.4,1-.4);%

\draw (1-.4,2-.4)--(1-.4,2.4)--(1.4,2.4)--(1.4,2-.4)--(1-.4,2-.4);%

\draw (2-.4,1-.4)--(2-.4,1.4)--(2.4,1.4)--(2.4,1-.4)--(2-.4,1-.4);%

\draw[ultra thick] (-5,-1)--(5,-1);

\end{scope}

\begin{scope}[xshift=10cm,yshift=-.5cm]

\node () at (0,3) {$\varphi_{\textcolor{blue}{a}} = \textcolor{blue}{\bullet}$};

\node () at (0,2.5) {$\varphi_{\textcolor{blue}{b}} = \textcolor{blue}{\bullet}$};

\node () at (0,2) {$\varphi_{\textcolor{blue}{c}} = \textcolor{blue}{\bullet}$};

\node () at (0,1.5) {$\varphi_{\textcolor{green}{d}} = \textcolor{green}{\bullet}$};

\node () at (0,1) {$\varphi_{\textcolor{pink}{e}} = \textcolor{pink}{\bullet}$};

\node () at (0,.5) {$\varphi_{\textcolor{yellow}{f}} = \textcolor{yellow}{\bullet}$};

\node () at (0,0) {$\varphi_{\textcolor{blue}{g}} = \textcolor{blue}{\bullet}$};

\end{scope}

\begin{scope}[yshift=-4cm]

    \node[place] (a) at (0,0) {$a$};
    \node[place] (b) at (0,1) {$b$};
    \node[place] (c) at (0,2) {$c$};
    \node[place] (d) at (1,0) {$d$};
    \node[square] (ef) at (1,1.5) {$ef$}
        edge[style=dashed,color=red,ultra thick] (a)
        edge[style=dashed,color=red,ultra thick] (d);
    \node[place] (g) at (2,1) {$g$};

    \draw (ef)--(g)--(d)--(a)--(b)--(c)--(ef)--(b)--(d);

\end{scope}

\begin{scope}[xshift=5cm,yshift=-4cm]

    \node[place,fill=blue] (a) at (0,0) {$a$};
    \node[place,fill=blue] (b) at (0,1) {$b$};
    \node[place,fill=blue] (c) at (0,2) {$c$};
    \node[place,fill=green] (d) at (1,0) {$d$};
    \node[place,fill=pink] (e) at (1,1) {$e$};
    \node[place,fill=pink] (f) at (1,2) {$f$};
    \node[place,fill=blue] (g) at (2,1) {$g$};

    \draw (e)--(d)--(a)--(f);

\draw (0.6,2.4)--(1.4,2.4)--(1.4,-0.4)--(-0.4,-0.4)--(-0.4,0.4)--(0.6,0.4)--(0.6,2.4);

\draw (-.4,1-.4)--(-.4,1.4)--(.4,1.4)--(.4,1-.4)--(-.4,1-.4);%

\draw (-.4,2-.4)--(-.4,2.4)--(.4,2.4)--(.4,2-.4)--(-.4,2-.4);%

\draw (2-.4,1-.4)--(2-.4,1.4)--(2.4,1.4)--(2.4,1-.4)--(2-.4,1-.4);%

\end{scope}

\begin{scope}[xshift=10cm,yshift=-4cm]

\node () at (0,2) {$\varphi_{\textcolor{blue}{a}\textcolor{green}{d}\textcolor{pink}{e}\textcolor{pink}{f}} = $};

\node () at (0,1) {$\rho_{\textcolor{yellow}{\bullet}\rightarrow \textcolor{pink}{\bullet}}$};

\node () at (0,.5) {$\eta_{\textcolor{pink}{\bullet},\textcolor{green}{\bullet}} \eta_{\textcolor{green}{\bullet},\textcolor{blue}{\bullet}} \eta_{\textcolor{blue}{\bullet},\textcolor{yellow}{\bullet}}$};

\node () at (0,0) {$(\varphi_{\textcolor{blue}{a}}\oplus \varphi_{\textcolor{green}{d}} \oplus \varphi_{\textcolor{pink}{e}} \oplus \varphi_{\textcolor{yellow}{f}})$};

\end{scope}

\end{tikzpicture}

}

\caption{We adapt the labels within components in anticipation of the contraction, perform disjoint unions, and construct the correct edges. Then, we set $e$ and $f$ to the same color.}
\end{figure}

\begin{figure}[H]
\centering

\scalebox{1.2}{

\begin{tikzpicture}
\tikzstyle{place}=[draw,shape=circle];
\tikzstyle{square}=[draw,shape=rectangle];

\begin{scope}

    \node[place] (a) at (0,0) {$a$};
    \node[place] (b) at (0,1) {$b$};
    \node[place] (c) at (0,2) {$c$};
    \node[place] (d) at (1,0) {$d$};
    \node[square] (ef) at (1,1.5) {$ef$}
        edge[style=dashed,color=red,ultra thick] (a)
        edge[style=dashed,color=red,ultra thick] (d);
    \node[place] (g) at (2,1) {$g$};

    \draw (ef)--(g)--(d)--(a)--(b)--(c)--(ef)--(b)--(d);

\end{scope}

\begin{scope}[xshift=5cm]

    \node[place,fill=blue] (a) at (0,0) {$a$};
    \node[place,fill=blue] (b) at (0,1) {$b$};
    \node[place,fill=blue] (c) at (0,2) {$c$};
    \node[place,fill=green] (d) at (1,0) {$d$};
    \node[place,fill=pink] (e) at (1,1) {$e$};
    \node[place,fill=pink] (f) at (1,2) {$f$};
    \node[place,fill=yellow] (g) at (2,1) {$g$};

    \draw (e)--(d)--(a)--(f);

\draw (0.6,2.4)--(1.4,2.4)--(1.4,-0.4)--(-0.4,-0.4)--(-0.4,0.4)--(0.6,0.4)--(0.6,2.4);

\draw (-.4,1-.4)--(-.4,1.4)--(.4,1.4)--(.4,1-.4)--(-.4,1-.4);%

\draw (-.4,2-.4)--(-.4,2.4)--(.4,2.4)--(.4,2-.4)--(-.4,2-.4);%

\draw (2-.4,1-.4)--(2-.4,1.4)--(2.4,1.4)--(2.4,1-.4)--(2-.4,1-.4);%

\draw[ultra thick] (-5,-1)--(5,-1);

\end{scope}

\begin{scope}[xshift=10cm]

\node () at (0,1) {$\varphi_{\textcolor{blue}{a}\textcolor{green}{d}\textcolor{pink}{e}\textcolor{pink}{f}}$};

\node () at (0,.5) {$\varphi_{\mathcolor{yellow}{g}}={\mathcolor{yellow}{\bullet}}$};

\end{scope}

\begin{scope}[yshift=-4cm]

    \node[square] (ad) at (1,0) {$ad$};
    \node[place] (b) at (0,1) {$b$};
    \node[place] (c) at (0,2) {$c$};
    \node[square] (ef) at (1,1.5) {$ef$}
        edge[style=dashed,color=red,ultra thick] (ad);
    \node[place] (g) at (2,1) {$g$}
        edge[style=dashed,color=red,ultra thick] (ad);

    \draw (ad)--(b)--(c)--(ef)--(b)--(ad);
    \draw (g)--(ef);
    
\end{scope}

\begin{scope}[xshift=5cm,yshift=-4cm]

    \node[place,fill=green] (a) at (0,0) {$a$};
    \node[place,fill=blue] (b) at (0,1) {$b$};
    \node[place,fill=blue] (c) at (0,2) {$c$};
    \node[place,fill=green] (d) at (1,0) {$d$};
    \node[place,fill=pink] (e) at (1,1) {$e$};
    \node[place,fill=pink] (f) at (1,2) {$f$};
    \node[place,fill=yellow] (g) at (2,1) {$g$};

    \draw (e)--(d)--(a)--(f);

    \draw (f)--(g)--(e)--(g)--(d);

\draw (0.6,2.4)--(2.4,2.4)--(2.4,-0.4)--(-0.4,-0.4)--(-0.4,0.4)--(0.6,0.4)--(0.6,2.4);%

\draw (-.4,1-.4)--(-.4,1.4)--(.4,1.4)--(.4,1-.4)--(-.4,1-.4);%

\draw (-.4,2-.4)--(-.4,2.4)--(.4,2.4)--(.4,2-.4)--(-.4,2-.4);%

\end{scope}

\begin{scope}[xshift=10cm,yshift=-4cm]

\node () at (0,1.5) {$\varphi_{\textcolor{green}{ad}\textcolor{pink}{ef}\textcolor{yellow}{g}} = $};

\node () at (0,1) {$\rho_{\textcolor{blue}{\bullet}\rightarrow\textcolor{green}{\bullet}}$};

\node () at (0,.5) {$\eta_{\textcolor{yellow}{\bullet},\textcolor{green}{\bullet}}\eta_{\textcolor{yellow}{\bullet},\textcolor{pink}{\bullet}}$};

\node () at (0,0) {$(\varphi_{\textcolor{blue}{a}\textcolor{green}{d}\textcolor{pink}{e}\textcolor{pink}{f}}\oplus \varphi_{\textcolor{yellow}{g}})$};

\end{scope}

\end{tikzpicture}

}

\caption{Now, $g$ joins the big red-connected component. Crucially, $e$ and $f$ ``agree'' on $g$.}
\end{figure}

\begin{figure}[H]
\centering

\scalebox{1.2}{

\begin{tikzpicture}
\tikzstyle{place}=[draw,shape=circle];
\tikzstyle{square}=[draw,shape=rectangle];

\begin{scope}

    \node[square] (ad) at (1,0) {$ad$};
    \node[place] (b) at (0,1) {$b$};
    \node[place] (c) at (0,2) {$c$};
    \node[square] (ef) at (1,1.5) {$ef$}
        edge[style=dashed,color=red,ultra thick] (ad);
    \node[place] (g) at (2,1) {$g$}
        edge[style=dashed,color=red,ultra thick] (ad);

    \draw (ad)--(b)--(c)--(ef)--(b)--(ad);
    \draw (g)--(ef);
    
\end{scope}

\begin{scope}[xshift=5cm]

    \node[place,fill=green] (a) at (0,0) {$a$};
    \node[place,fill=blue] (b) at (0,1) {$b$};
    \node[place,fill=blue] (c) at (0,2) {$c$};
    \node[place,fill=green] (d) at (1,0) {$d$};
    \node[place,fill=pink] (e) at (1,1) {$e$};
    \node[place,fill=pink] (f) at (1,2) {$f$};
    \node[place,fill=yellow] (g) at (2,1) {$g$};

    \draw (e)--(d)--(a)--(f);

    \draw (f)--(g)--(e)--(g)--(d);

\draw (0.6,2.4)--(2.4,2.4)--(2.4,-0.4)--(-0.4,-0.4)--(-0.4,0.4)--(0.6,0.4)--(0.6,2.4);%

\draw (-.4,1-.4)--(-.4,1.4)--(.4,1.4)--(.4,1-.4)--(-.4,1-.4);%

\draw (-.4,2-.4)--(-.4,2.4)--(.4,2.4)--(.4,2-.4)--(-.4,2-.4);%

\draw[ultra thick] (-5,-1)--(5,-1);

\end{scope}

\begin{scope}[xshift=10cm]

\node () at (0,1.5) {$\varphi_{\textcolor{green}{ad}\textcolor{pink}{ef}\textcolor{yellow}{g}} = $};

\node () at (0,1) {$\varphi_{\textcolor{blue}{b}}=\mathcolor{blue}{\bullet}$};

\end{scope}

\begin{scope}[yshift=-4cm]

    \node[square] (ad) at (1,0) {$ad$};
    \node[place] (c) at (0,2) {$c$};
    \node[square] (bef) at (1,1.5) {$bef$}
        edge[style=dashed,color=red,ultra thick] (ad);
    \node[place] (g) at (2,1) {$g$}
        edge[style=dashed,color=red,ultra thick] (ad)
        edge[style=dashed,color=red,ultra thick] (bef);

    \draw (bef)--(c);

\end{scope}

\begin{scope}[xshift=5cm,yshift=-4cm]

    \node[place,fill=green] (a) at (0,0) {$a$};
    \node[place,fill=pink] (b) at (0,1) {$b$};
    \node[place,fill=blue] (c) at (0,2) {$c$};
    \node[place,fill=green] (d) at (1,0) {$d$};
    \node[place,fill=pink] (e) at (1,1) {$e$};
    \node[place,fill=pink] (f) at (1,2) {$f$};
    \node[place,fill=yellow] (g) at (2,1) {$g$};

    \draw (e)--(d)--(a)--(f);

    \draw (f)--(g)--(e)--(g)--(d);

    \draw (e)--(b)--(f);

    \draw (a)--(b)--(d);

\draw (0.6,2.4)--(2.4,2.4)--(2.4,-0.4)--(-0.4,-0.4)--(-0.4,1.4)--(0.6,1.4)--(0.6,2.4);%

\draw (-.4,2-.4)--(-.4,2.4)--(.4,2.4)--(.4,2-.4)--(-.4,2-.4);%

\end{scope}

\begin{scope}[xshift=10cm,yshift=-4cm]

\node () at (0,1.5) {$\varphi_{\textcolor{green}{ad}\textcolor{pink}{bef}\textcolor{yellow}{g}} = $};

\node () at (0,1) {$\rho_{\textcolor{blue}{\bullet}\rightarrow\textcolor{pink}{\bullet}}$};

\node () at (0,.5) {$\eta_{\textcolor{blue}{\bullet},\textcolor{green}{\bullet}}\eta_{\textcolor{blue}{\bullet},\textcolor{pink}{\bullet}}$};

\node () at (0,0) {$(\varphi_{\textcolor{green}{ad}\textcolor{pink}{ef}\textcolor{yellow}{g}}\oplus \varphi_{\textcolor{blue}{b}})$};

\end{scope}

\end{tikzpicture}

}

\caption{$b$ joins the ``big'' red-connected component}
\end{figure}

\begin{figure}[H]
\centering

\scalebox{1.2}{

\begin{tikzpicture}
\tikzstyle{place}=[draw,shape=circle];
\tikzstyle{square}=[draw,shape=rectangle];

\begin{scope}

    \node[square] (ad) at (1,0) {$ad$};
    \node[place] (c) at (0,2) {$c$};
    \node[square] (bef) at (1,1.5) {$bef$}
        edge[style=dashed,color=red,ultra thick] (ad);
    \node[place] (g) at (2,1) {$g$}
        edge[style=dashed,color=red,ultra thick] (ad)
        edge[style=dashed,color=red,ultra thick] (bef);

    \draw (bef)--(c);

\end{scope}

\begin{scope}[xshift=5cm]

    \node[place,fill=green] (a) at (0,0) {$a$};
    \node[place,fill=pink] (b) at (0,1) {$b$};
    \node[place,fill=blue] (c) at (0,2) {$c$};
    \node[place,fill=green] (d) at (1,0) {$d$};
    \node[place,fill=pink] (e) at (1,1) {$e$};
    \node[place,fill=pink] (f) at (1,2) {$f$};
    \node[place,fill=yellow] (g) at (2,1) {$g$};

    \draw (e)--(d)--(a)--(f);

    \draw (f)--(g)--(e)--(g)--(d);

    \draw (e)--(b)--(f);

    \draw (a)--(b)--(d);

\draw (0.6,2.4)--(2.4,2.4)--(2.4,-0.4)--(-0.4,-0.4)--(-0.4,1.4)--(0.6,1.4)--(0.6,2.4);%

\draw (-.4,2-.4)--(-.4,2.4)--(.4,2.4)--(.4,2-.4)--(-.4,2-.4);%

\draw[ultra thick] (-5,-1)--(5,-1);
   
\end{scope}

\begin{scope}[xshift=10cm]

\node () at (0,1) {$\varphi_{\textcolor{green}{ad}\textcolor{pink}{bef}\textcolor{yellow}{g}}$};

\end{scope}

\begin{scope}[yshift=-3cm]

    \node[place] (c) at (0,0) {$c$};
    \node[square] (bef) at (1.5,0) {$bef$};
    \node[square] (adg) at (3,0) {$adg$}
        edge[style=dashed,color=red,ultra thick] (bef);

    \draw (bef)--(c);

\end{scope}

\begin{scope}[xshift=5cm,yshift=-4cm]

    \node[place,fill=green] (a) at (0,0) {$a$};
    \node[place,fill=pink] (b) at (0,1) {$b$};
    \node[place,fill=blue] (c) at (0,2) {$c$};
    \node[place,fill=green] (d) at (1,0) {$d$};
    \node[place,fill=pink] (e) at (1,1) {$e$};
    \node[place,fill=pink] (f) at (1,2) {$f$};
    \node[place,fill=green] (g) at (2,1) {$g$};

    \draw (e)--(d)--(a)--(f);

    \draw (f)--(g)--(e)--(g)--(d);

    \draw (e)--(b)--(f);

    \draw (a)--(b)--(d);

\draw (0.6,2.4)--(2.4,2.4)--(2.4,-0.4)--(-0.4,-0.4)--(-0.4,1.4)--(0.6,1.4)--(0.6,2.4);%

\draw (-.4,2-.4)--(-.4,2.4)--(.4,2.4)--(.4,2-.4)--(-.4,2-.4);%

\end{scope}

\begin{scope}[xshift=10cm,yshift=-3.5cm]

\node () at (0,1) {$\varphi_{\textcolor{green}{adg}\textcolor{pink}{bef}} = $};

\node () at (0,.5) {$\rho_{\textcolor{yellow}{\bullet}\rightarrow\textcolor{green}{\bullet}}$};

\node () at (0,0) {$\varphi_{\textcolor{green}{ad}\textcolor{pink}{bef}\textcolor{yellow}{g}}$};

\end{scope}

\end{tikzpicture}

}

\caption{The red-components are the same: only a relabelling happens.}
\end{figure}

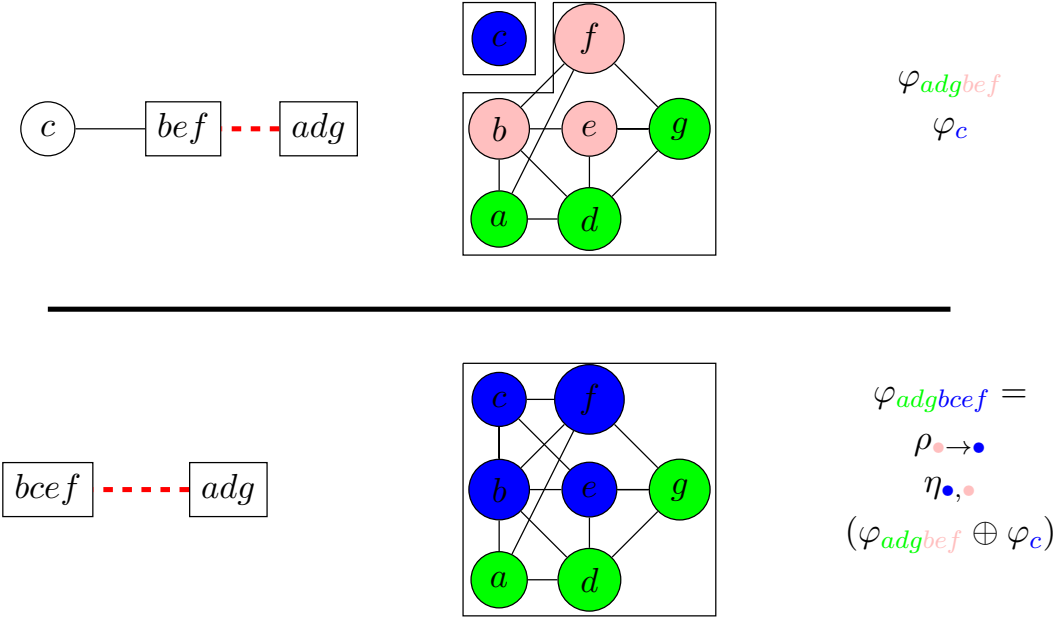
\begin{figure}[H]
\centering

\scalebox{1.2}{

\begin{tikzpicture}
\tikzstyle{place}=[draw,shape=circle];
\tikzstyle{square}=[draw,shape=rectangle];

\begin{scope}[yshift=1cm]

    \node[place] (c) at (0,0) {$c$};
    \node[square] (bef) at (1.5,0) {$bef$};
    \node[square] (adg) at (3,0) {$adg$}
        edge[style=dashed,color=red,ultra thick] (bef);

    \draw (bef)--(c);

\end{scope}

\begin{scope}[xshift=5cm]

    \node[place,fill=green] (a) at (0,0) {$a$};
    \node[place,fill=pink] (b) at (0,1) {$b$};
    \node[place,fill=blue] (c) at (0,2) {$c$};
    \node[place,fill=green] (d) at (1,0) {$d$};
    \node[place,fill=pink] (e) at (1,1) {$e$};
    \node[place,fill=pink] (f) at (1,2) {$f$};
    \node[place,fill=green] (g) at (2,1) {$g$};

    \draw (e)--(d)--(a)--(f);

    \draw (f)--(g)--(e)--(g)--(d);

    \draw (e)--(b)--(f);

    \draw (a)--(b)--(d);

\draw (0.6,2.4)--(2.4,2.4)--(2.4,-0.4)--(-0.4,-0.4)--(-0.4,1.4)--(0.6,1.4)--(0.6,2.4);%

\draw (-.4,2-.4)--(-.4,2.4)--(.4,2.4)--(.4,2-.4)--(-.4,2-.4);%

\draw[ultra thick] (-5,-1)--(5,-1);

\end{scope}

\begin{scope}[xshift=10cm,yshift=.5cm]

\node () at (0,1) {$\varphi_{\textcolor{green}{adg}\textcolor{pink}{bef}}$};

\node () at (0,.5) {$\varphi_{\mathcolor{blue}{c}}$};

\end{scope}

\begin{scope}[yshift=-3cm]

    \node[square] (bcef) at (0,0) {$bcef$};
    \node[square] (adg) at (2,0) {$adg$}
        edge[style=dashed,color=red,ultra thick] (bcef);

\end{scope}

\begin{scope}[xshift=5cm,yshift=-4cm]

    \node[place,fill=green] (a) at (0,0) {$a$};
    \node[place,fill=blue] (b) at (0,1) {$b$};
    \node[place,fill=blue] (c) at (0,2) {$c$};
    \node[place,fill=green] (d) at (1,0) {$d$};
    \node[place,fill=blue] (e) at (1,1) {$e$};
    \node[place,fill=blue] (f) at (1,2) {$f$};
    \node[place,fill=green] (g) at (2,1) {$g$};

    \draw (e)--(d)--(a)--(f);

    \draw (f)--(g)--(e)--(g)--(d);

    \draw (e)--(b)--(f);

    \draw (a)--(b)--(d);

    \draw (f)--(c)--(b)--(c)--(e);

    \draw (-0.4,2.4)--(2.4,2.4)--(2.4,-0.4)--(-0.4,-0.4)--(-0.4,2.4);

\end{scope}

\begin{scope}[xshift=10cm,yshift=-3.5cm]

\node () at (0,1.5) {$\varphi_{\textcolor{green}{adg}\textcolor{blue}{bcef}}=$};

\node () at (0,1) {$\rho_{\textcolor{pink}{\bullet}\rightarrow\textcolor{blue}{\bullet}}$};

\node () at (0,.5) {$\eta_{\textcolor{blue}{\bullet},\textcolor{pink}{\bullet}}$};

\node () at (0,0) {$(\varphi_{\textcolor{green}{adg}\textcolor{pink}{bef}}\oplus\varphi_{\textcolor{blue}{c}})$};

\end{scope}

\end{tikzpicture}

}

\caption{Now $c$ joins the ``big component''. We already have a $4$-expression of the original graph.}
\end{figure}

\begin{figure}[H]
\centering

\scalebox{1.2}{

\begin{tikzpicture}
\tikzstyle{place}=[draw,shape=circle];
\tikzstyle{square}=[draw,shape=rectangle];

\begin{scope}[yshift=1cm]

    \node () at (-1,0) {};
    
    \node[square] (abcdefg) at (0,0) {$abcdefg$};

\end{scope}

\begin{scope}[xshift=4.4cm]

    \node[place,fill=blue] (a) at (0,0) {$a$};
    \node[place,fill=blue] (b) at (0,1) {$b$};
    \node[place,fill=blue] (c) at (0,2) {$c$};
    \node[place,fill=blue] (d) at (1,0) {$d$};
    \node[place,fill=blue] (e) at (1,1) {$e$};
    \node[place,fill=blue] (f) at (1,2) {$f$};
    \node[place,fill=blue] (g) at (2,1) {$g$};

    \draw (e)--(d)--(a)--(f);

    \draw (f)--(g)--(e)--(g)--(d);

    \draw (e)--(b)--(f);

    \draw (a)--(b)--(d);

    \draw (f)--(c)--(b)--(c)--(e);

    \draw (-0.4,2.4)--(2.4,2.4)--(2.4,-0.4)--(-0.4,-0.4)--(-0.4,2.4);

\end{scope}

\begin{scope}[xshift=9.5cm]

\node () at (0,1) {$\varphi_{\textcolor{blue}{abcdefg}}$};

\end{scope}

\end{tikzpicture}

}

\caption{Final situation.}
\end{figure}
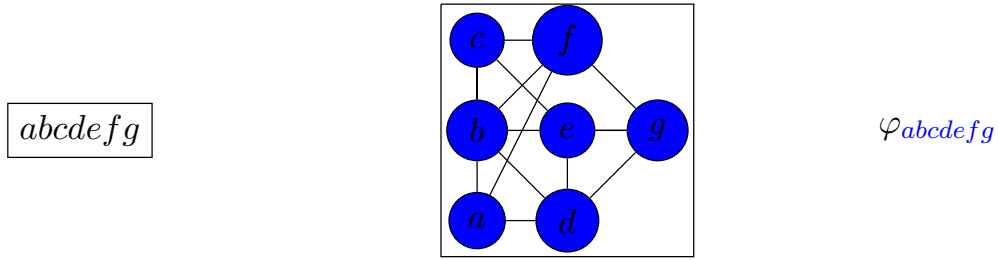

\subsection{From $k$-expressions to contraction sequences}\label{app:ctww vs cw}

We continue by illustrating an example of the application of the method described in Lemma \ref{lem:ctww leq 2cw-1} (establishing the rightmost part of the linear bounds of Theorem \ref{thm:cw vs ctww}).

\SecondLinearBound*

We concentrate on illustrating $(i)$, since $(ii)$ is analogous.
The method takes as an input a $k$-expression of a graph, and uses it to describe a contraction sequence of the same graph that witnesses that its component twin-width is $\le 2k$.

This method progressively ``collapses''  the $k$-expression. A partition of the vertices of the original graph correspond naturally to every step of the collapse: two vertices are in the same subset of the partition if they have been collapsed together. Subsets of vertices that have been collapsed together are referred to as {\em parks}.

\begin{figure}[H]
\centering

\scalebox{1.1}{

\begin{tikzpicture}[grow=up,every tree node/.style={draw},
   level distance=0.8cm,sibling distance=.05cm,
   edge from parent/.style={draw,-latex,<-}]
\tikzstyle{place}=[draw,shape=circle];
\tikzstyle{square}=[draw,shape=rectangle];

\begin{scope}

\Tree
[.\node{$\eta_{\textcolor{green}{\bullet},\textcolor{pink}{\bullet}}$};
    [.\node{$\oplus$};
        [.\node{$\rho_{\textcolor{pink}{\bullet}\rightarrow\textcolor{green}{\bullet}}$};
            [.\node{$\eta_{\textcolor{blue}{\bullet},\textcolor{green}{\bullet}}$};
                [.\node{$\oplus$};
                    [.\node[ultra thick]{$\textcolor{pink}{j}$};]
                    [.\node{$\eta_{\textcolor{blue}{\bullet},\textcolor{green}{\bullet}}$};
                        [.\node{$\oplus$};
                            [.\node{$\oplus$};
                                [.\node[ultra thick]{$\textcolor{blue}{i}$};]
                                [.\node[ultra thick]{$\textcolor{blue}{h}$};]
                            ]
                            [.\node[ultra thick]{$\textcolor{green}{g}$};]
                        ]
                    ]
                ]
            ]
        ]
        [.\node{$\rho_{\textcolor{green}{\bullet}\rightarrow\textcolor{blue}{\bullet}}$};
            [.\node{$\eta_{\textcolor{blue}{\bullet},\textcolor{pink}{\bullet}}$};
                [.\node{$\oplus$};
                    [.\node{$\oplus$};
                        [.\node{$\oplus$};
                            [.\node{$\oplus$};
                                [.\node[ultra thick]{$\textcolor{pink}{f}$};]
                                [.\node[ultra thick]{$\textcolor{pink}{e}$};]
                            ]
                            [.\node[ultra thick]{$\textcolor{pink}{d}$};]
                        ]
                        [.\node[ultra thick]{$\textcolor{blue}{c}$};]
                     ]
                     [.\node{$\eta_{\textcolor{blue}{\bullet},\textcolor{green}{\bullet}}$};
                        [.\node{$\oplus$};
                            [.\node[ultra thick]{$\textcolor{blue}{b}$};
                            ]
                            [.\node[ultra thick]{$\textcolor{green}{a}$};
                            ]
                        ]
                     ]
                ]
            ]
        ]
    ]
]

\end{scope}

\begin{scope}[xshift=5cm,yshift=2cm]

    \node[place,fill=green] (a) at (0,0) {$a$};
    \node[place,fill=blue] (b) at (0,1) {$b$};
    \node[place,fill=blue] (c) at (0,2) {$c$};
    \node[place,fill=pink] (d) at (1.5,0) {$d$};
    \node[place,fill=pink] (e) at (1.5,1) {$e$};
    \node[place,fill=pink] (f) at (1.5,2) {$f$};
    \node[place,fill=green] (g) at (3,1.5) {$g$};
    \node[place,fill=blue] (h) at (4.5,2) {$h$};
    \node[place,fill=blue] (i) at (4.5,1) {$i$};
    \node[place,fill=pink] (j) at (3,0.5) {$j$};

\draw (h)--(g)--(f)--(c)--(e)--(g)--(i);

\draw (g)--(j)--(e)--(b)--(d)--(j);

\draw (c)--(d);

\draw (a)--(b)--(f);

\draw[ultra thick] (-.4,-.4)--(+.4,-.4)--(+.4,+.4)--(-.4,+.4)--(-.4,-.4);

\draw[ultra thick] (-.4,1-.4)--(+.4,1-.4)--(+.4,1+.4)--(-.4,1+.4)--(-.4,1-.4);

\draw[ultra thick] (-.4,2-.4)--(+.4,2-.4)--(+.4,2+.4)--(-.4,2+.4)--(-.4,2-.4);

\draw[ultra thick] (1.5-.4,-.4)--(1.5+.4,-.4)--(1.5+.4,+.4)--(1.5-.4,+.4)--(1.5-.4,-.4);

\draw[ultra thick] (1.5-.4,1-.4)--(1.5+.4,1-.4)--(1.5+.4,1+.4)--(1.5-.4,1+.4)--(1.5-.4,1-.4);

\draw[ultra thick] (1.5-.4,2-.4)--(1.5+.4,2-.4)--(1.5+.4,2+.4)--(1.5-.4,2+.4)--(1.5-.4,2-.4);

\draw[ultra thick] (3-.4,.5-.4)--(3+.4,.5-.4)--(3+.4,.5+.4)--(3-.4,.5+.4)--(3-.4,.5-.4);

\draw[ultra thick] (3-.4,1.5-.4)--(3+.4,1.5-.4)--(3+.4,1.5+.4)--(3-.4,1.5+.4)--(3-.4,1.5-.4);

\draw[ultra thick] (4.5-.4,2-.4)--(4.5+.4,2-.4)--(4.5+.4,2+.4)--(4.5-.4,2+.4)--(4.5-.4,2-.4);

\draw[ultra thick] (4.5-.4,1-.4)--(4.5+.4,1-.4)--(4.5+.4,1+.4)--(4.5-.4,1+.4)--(4.5-.4,1-.4);

\end{scope}

\end{tikzpicture}

}

\caption{We represent the vertices with their current label.}
\end{figure}
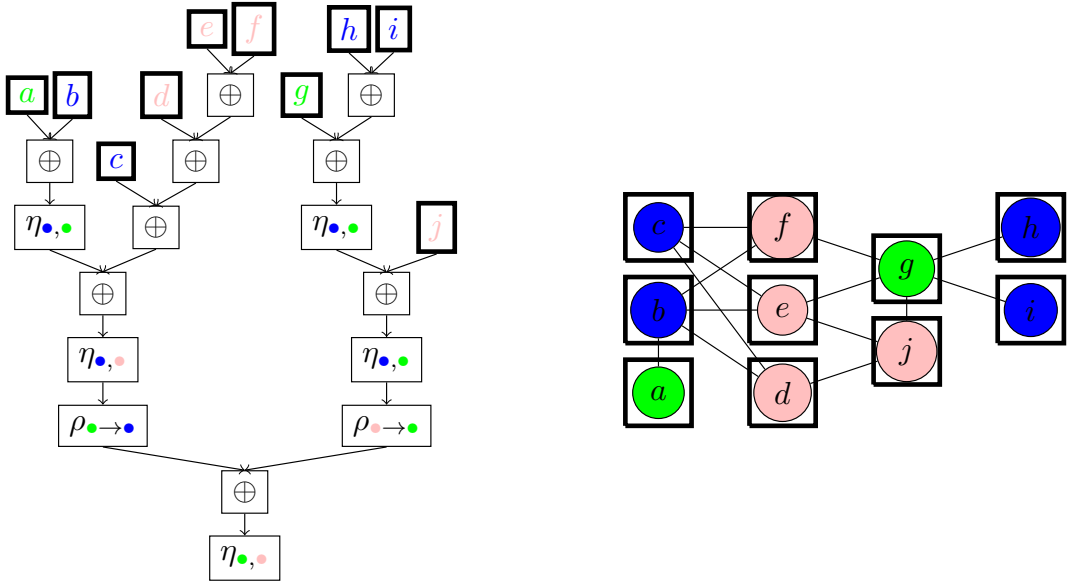

\begin{figure}[H]
\centering

\scalebox{1.1}{

\begin{tikzpicture}[grow=up,every tree node/.style={draw},
   level distance=0.8cm,sibling distance=.05cm,
   edge from parent/.style={draw,-latex,<-}]
\tikzstyle{place}=[draw,shape=circle];
\tikzstyle{square}=[draw,shape=rectangle];

\begin{scope}

\Tree
[.\node{$\eta_{\textcolor{green}{\bullet},\textcolor{pink}{\bullet}}$};
    [.\node{$\oplus$};
        [.\node{$\rho_{\textcolor{pink}{\bullet}\rightarrow\textcolor{green}{\bullet}}$};
            [.\node{$\eta_{\textcolor{blue}{\bullet},\textcolor{green}{\bullet}}$};
                [.\node{$\oplus$};
                    [.\node[ultra thick]{$\textcolor{pink}{j}$};]
                    [.\node{$\eta_{\textcolor{blue}{\bullet},\textcolor{green}{\bullet}}$};
                        [.\node{$\oplus$};
                            [.\node[ultra thick]{$\textcolor{blue}{hi}$};]
                            [.\node[ultra thick]{$\textcolor{green}{g}$};]
                        ]
                    ]
                ]
            ]
        ]
        [.\node{$\rho_{\textcolor{green}{\bullet}\rightarrow\textcolor{blue}{\bullet}}$};
            [.\node{$\eta_{\textcolor{blue}{\bullet},\textcolor{pink}{\bullet}}$};
                [.\node{$\oplus$};
                    [.\node{$\oplus$};
                        [.\node[ultra thick]{$\textcolor{pink}{def}$};][.\node[ultra thick]{$\textcolor{blue}{c}$};]
                    ]
                     [.\node{$\eta_{\textcolor{blue}{\bullet},\textcolor{green}{\bullet}}$};
                        [.\node{$\oplus$};
                            [.\node[ultra thick]{$\textcolor{blue}{b}$};
                            ]
                            [.\node[ultra thick]{$\textcolor{green}{a}$};
                            ]
                        ]
                     ]
                ]
            ]
        ]
    ]
]

\end{scope}

\begin{scope}[xshift=5cm,yshift=2cm]

    \node[place,fill=green] (a) at (0,0) {$a$};
    \node[place,fill=blue] (b) at (0,1) {$b$};
    \node[place,fill=blue] (c) at (0,2) {$c$};
    \node[place,fill=pink] (def) at (1.5,2) {$def$};
    \node[place,fill=green] (g) at (3,1.5) {$g$};
    \node[place,fill=blue] (hi) at (4.5,1.5) {$hi$};
    \node[place,fill=pink] (j) at (3,0.5) {$j$};

\draw (hi)--(g)--(def)--(c)--(def)--(g)--(hi);

\draw (g)--(j)--(def)--(b)--(def)--(j);

\draw (c)--(def);

\draw (a)--(b)--(def);

\draw[ultra thick] (-.4,-.4)--(+.4,-.4)--(+.4,+.4)--(-.4,+.4)--(-.4,-.4);

\draw[ultra thick] (-.4,1-.4)--(+.4,1-.4)--(+.4,1+.4)--(-.4,1+.4)--(-.4,1-.4);

\draw[ultra thick] (-.4,2-.4)--(+.4,2-.4)--(+.4,2+.4)--(-.4,2+.4)--(-.4,2-.4);

\draw[ultra thick] (1.5-.6,2-.6)--(1.5+.6,2-.6)--(1.5+.6,2+.6)--(1.5-.6,2+.6)--(1.5-.6,2-.6);%

\draw[ultra thick] (3-.4,.5-.4)--(3+.4,.5-.4)--(3+.4,.5+.4)--(3-.4,.5+.4)--(3-.4,.5-.4);

\draw[ultra thick] (3-.4,1.5-.4)--(3+.4,1.5-.4)--(3+.4,1.5+.4)--(3-.4,1.5+.4)--(3-.4,1.5-.4);

\draw[ultra thick] (4.5-.5,1.5-.5)--(4.5+.5,1.5-.5)--(4.5+.5,1.5+.5)--(4.5-.5,1.5+.5)--(4.5-.5,1.5-.5);

\end{scope}

\end{tikzpicture}

}

\caption{$d$, $e$ and $f$ are introduced together with the same label: they are twins. We can contract them.}
\end{figure}
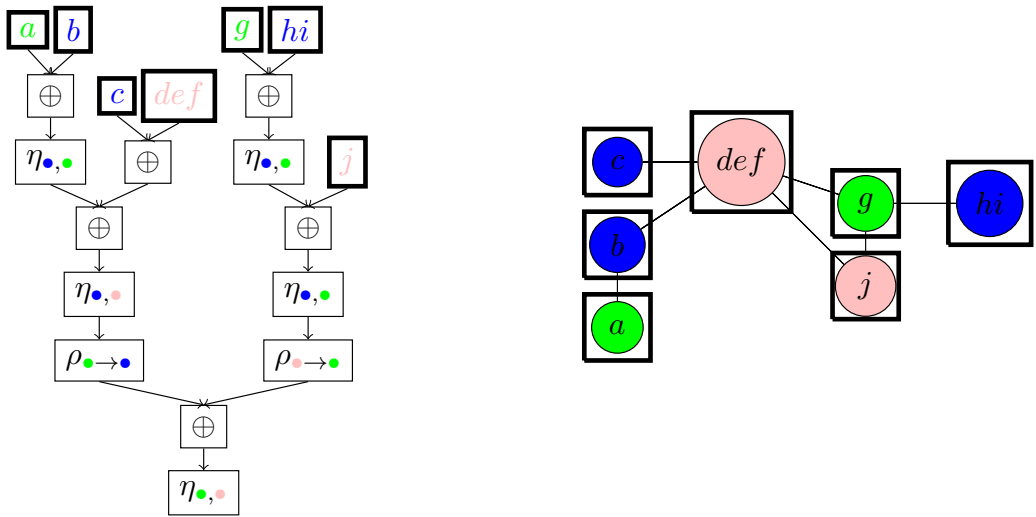

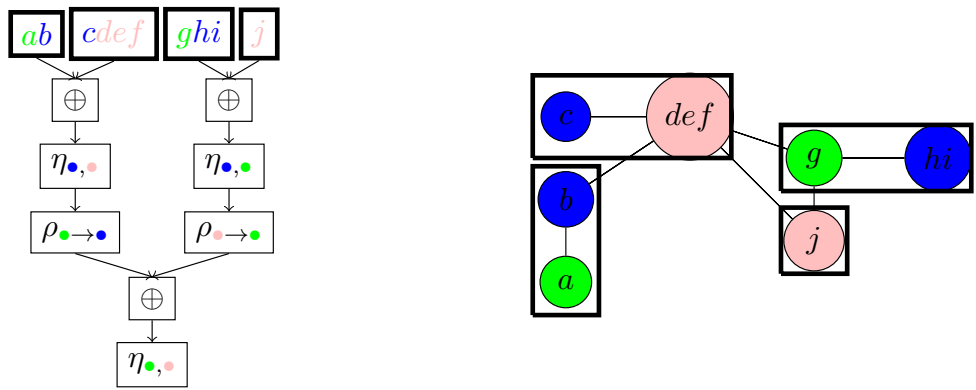
\begin{figure}[H]
\centering

\scalebox{1.1}{

\begin{tikzpicture}[grow=up,every tree node/.style={draw},
   level distance=0.8cm,sibling distance=.05cm,
   edge from parent/.style={draw,-latex,<-}]
\tikzstyle{place}=[draw,shape=circle];
\tikzstyle{square}=[draw,shape=rectangle];

\begin{scope}

\Tree
[.\node{$\eta_{\textcolor{green}{\bullet},\textcolor{pink}{\bullet}}$};
    [.\node{$\oplus$};
        [.\node{$\rho_{\textcolor{pink}{\bullet}\rightarrow\textcolor{green}{\bullet}}$};
            [.\node{$\eta_{\textcolor{blue}{\bullet},\textcolor{green}{\bullet}}$};
                [.\node{$\oplus$};
                    [.\node[ultra thick]{$\textcolor{pink}{j}$};]
                    [.\node[ultra thick]{$\textcolor{green}{g}\textcolor{blue}{hi}$};]
                ]
            ]
        ]
        [.\node{$\rho_{\textcolor{green}{\bullet}\rightarrow\textcolor{blue}{\bullet}}$};
            [.\node{$\eta_{\textcolor{blue}{\bullet},\textcolor{pink}{\bullet}}$};
                [.\node{$\oplus$};
                    [.\node[ultra thick]{$\textcolor{blue}{c}\textcolor{pink}{def}$};]
                    [.\node[ultra thick]{$\textcolor{green}{a}\textcolor{blue}{b}$};
                    ]
                ]
            ]
        ]
    ]
]

\end{scope}

\begin{scope}[xshift=5cm,yshift=2cm]

    \node[place,fill=green] (a) at (0,-1) {$a$};
    \node[place,fill=blue] (b) at (0,0) {$b$};
    \node[place,fill=blue] (c) at (0,1) {$c$};
    \node[place,fill=pink] (def) at (1.5,1) {$def$};
    \node[place,fill=green] (g) at (3,.5) {$g$};
    \node[place,fill=blue] (hi) at (4.5,.5) {$hi$};
    \node[place,fill=pink] (j) at (3,-.5) {$j$};

\draw (hi)--(g)--(def)--(c)--(def)--(g)--(hi);

\draw (g)--(j)--(def)--(b)--(def)--(j);

\draw (c)--(def);

\draw (a)--(b)--(def);

\draw[ultra thick] (-.4,-1-.4)--(+.4,-1-.4)--(+.4,.4)--(-.4,.4)--(-.4,-1-.4);%

\draw[ultra thick] (-.4,1-.5)--(1.5+.5,1-.5)--(1.5+.5,1+.5)--(-.4,1+.5)--(-.4,1-.5);%

\draw[ultra thick] (3-.4,-.5-.4)--(3+.4,-.5-.4)--(3+.4,-.5+.4)--(3-.4,-.5+.4)--(3-.4,-.5-.4);%

\draw[ultra thick] (3-.4,.5-.4)--(4.5+.4,.5-.4)--(4.5+.4,.5+.4)--(3-.4,.5+.4)--(3-.4,.5-.4); %

\end{scope}

\end{tikzpicture}

}

\caption{We collapse the $k$-expression and merge the ``parks'' accordingly.}
\end{figure}

\begin{figure}[H]
\centering

\scalebox{1.1}{

\begin{tikzpicture}[grow=up,every tree node/.style={draw},
   level distance=0.8cm,sibling distance=.05cm,
   edge from parent/.style={draw,-latex,<-}]
\tikzstyle{place}=[draw,shape=circle];
\tikzstyle{square}=[draw,shape=rectangle];

\begin{scope}

\Tree
[.\node{$\eta_{\textcolor{green}{\bullet},\textcolor{pink}{\bullet}}$};
    [.\node{$\oplus$};
        [.\node{$\rho_{\textcolor{pink}{\bullet}\rightarrow\textcolor{green}{\bullet}}$};
            [.\node[ultra thick]{$\textcolor{green}{g}\textcolor{blue}{hi}\textcolor{pink}{j}$};]
        ]
        [.\node{$\rho_{\textcolor{green}{\bullet}\rightarrow\textcolor{blue}{\bullet}}$};
            [.\node[ultra thick]{$\textcolor{green}{a}\textcolor{blue}{bc}\textcolor{pink}{def}$};
            ]
        ]
    ]
]

\end{scope}

\begin{scope}[xshift=5cm,yshift=1cm]

    \node[place,fill=green] (a) at (0,-.2) {$a$};
    \node[place,fill=blue] (bc) at (0,1) {$bc$};
    \node[place,fill=pink] (def) at (1.5,1) {$def$};
    \node[place,fill=green] (g) at (3,1) {$g$};
    \node[place,fill=blue] (hi) at (4.5,1) {$hi$};
    \node[place,fill=pink] (j) at (3,-.2) {$j$};

\draw (hi)--(g)--(def)--(bc)--(def)--(g)--(hi);

\draw (g)--(j)--(def)--(j);

\draw[ultra thick,red,style=dashed] (a)--(bc);

\draw[ultra thick] (-.5,-.2-.4)--(1.5+.6,-.2-.4)--(1.5+.6,1+.6)--(-.5,1+.6)--(-.5,-.2-.4);%

\draw[ultra thick] (3-.5,-.2-.4)--(4.5+.5,-.2-.4)--(4.5+.5,1+.6)--(3-.5,1+.6)--(3-.5,-.2-.4);

\end{scope}

\end{tikzpicture}

}

\caption{We merge vertices with the same label in the same park: the red-edges created are confined in the parks.}
\end{figure}
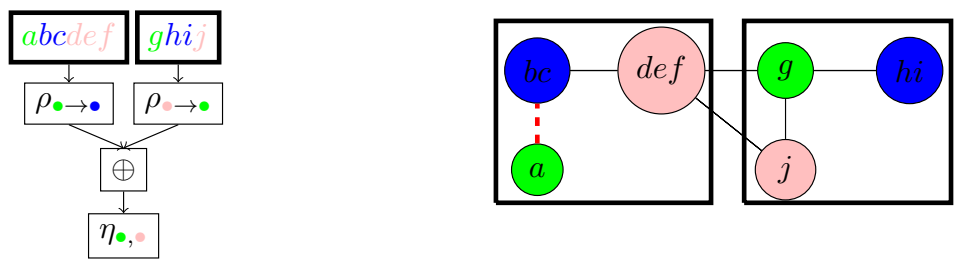

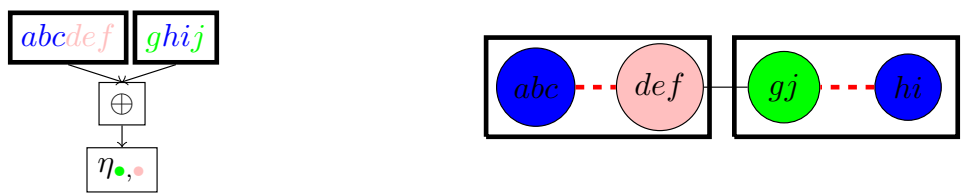
\begin{figure}[H]
\centering

\scalebox{1.1}{

\begin{tikzpicture}[grow=up,every tree node/.style={draw},
   level distance=0.8cm,sibling distance=.05cm,
   edge from parent/.style={draw,-latex,<-}]
\tikzstyle{place}=[draw,shape=circle];
\tikzstyle{square}=[draw,shape=rectangle];

\begin{scope}

\Tree
[.\node{$\eta_{\textcolor{green}{\bullet},\textcolor{pink}{\bullet}}$};
    [.\node{$\oplus$};
        [.\node[ultra thick]{$\textcolor{green}{g}\textcolor{blue}{hi}\textcolor{green}{j}$};]
        [.\node[ultra thick]{$\textcolor{blue}{abc}\textcolor{pink}{def}$};]
    ]
]

\end{scope}

\begin{scope}[xshift=5cm,yshift=0cm]

    \node[place,fill=blue] (abc) at (0,1) {$abc$};
    \node[place,fill=pink] (def) at (1.5,1) {$def$};
    \node[place,fill=green] (gj) at (3,1) {$gj$};
    \node[place,fill=blue] (hi) at (4.5,1) {$hi$};

\draw (gj)--(def);

\draw[ultra thick,red,style=dashed] (abc)--(def);

\draw[ultra thick,red,style=dashed] (hi)--(gj);

\draw[ultra thick] (-.6,1-.6)--(1.5+.6,1-.6)--(1.5+.6,1+.6)--(-.6,1+.6)--(-.6,1-.6);%

\draw[ultra thick] (3-.6,1-.6)--(4.5+.6,1-.6)--(4.5+.6,1+.6)--(3-.6,1+.6)--(3-.6,1-.6);

\end{scope}

\end{tikzpicture}

}

\caption{After the next step, only one park will remain. We can finish the contraction sequence arbiltrarly.}
\end{figure}

\begin{multicols}{2}

\begin{figure}[H]
\begin{tikzpicture}
\tikzstyle{place}=[draw,shape=circle];
\tikzstyle{square}=[draw,shape=rectangle];
    
    \node[place,fill=green] (a) at (0,-.2) {$2$};
    \node[place,fill=blue] (bc) at (0,1) {$1$};
    \node[place,fill=pink] (def) at (1.5,1) {$3$};
    \node[place,fill=green] (g) at (3,1) {$2'$};
    \node[place,fill=blue] (hi) at (4.5,1) {$1'$};
    \node[place,fill=pink] (j) at (3,-.2) {$3'$};

\draw (def)--(g);

\draw (j)--(a);

\draw[ultra thick,red,style=dashed] (a)--(bc)--(def)--(a);

\draw[ultra thick,red,style=dashed] (g)--(hi)--(j)--(g);

\draw[ultra thick] (-.5,-.2-.4)--(1.5+.6,-.2-.4)--(1.5+.6,1+.6)--(-.5,1+.6)--(-.5,-.2-.4);%

\draw[ultra thick] (3-.5,-.2-.4)--(4.5+.5,-.2-.4)--(4.5+.5,1+.6)--(3-.5,1+.6)--(3-.5,-.2-.4);

\end{tikzpicture}

\begin{tikzpicture}
\tikzstyle{place}=[draw,shape=circle];
\tikzstyle{square}=[draw,shape=rectangle];
    
    \node[place,fill=green] (a) at (0,-.2) {$2$};
    \node[place,fill=blue] (bc) at (0,1) {$1$};
    \node[place,fill=pink] (def) at (1.5,1) {$3$};
    \node[place,fill=green] (g) at (3,1) {$2'$};
    \node[place,fill=blue] (hi) at (4.5,1) {$1'$};
    \node[place,fill=pink] (j) at (3,-.2) {$3'$};

\draw (def)--(g);

\draw (j)--(a);

\draw[ultra thick,red,style=dashed] (a)--(bc)--(def)--(a);

\draw[ultra thick,red,style=dashed] (g)--(hi)--(j)--(g);

\draw[ultra thick] (-.5,-.2-.4)--(4.5+.5,-.2-.4)--(4.5+.5,1+.6)--(-.5,1+.6)--(-.5,-.2-.4);

\end{tikzpicture}

\begin{tikzpicture}
\tikzstyle{place}=[draw,shape=circle];
\tikzstyle{square}=[draw,shape=rectangle];
    
    \node[place,fill=green] (a) at (0,-.2) {$2$};
    \node[place,fill=pink] (def) at (0,1) {$3$};
    \node[place,fill=green] (g) at (4.5,1) {$2'$};
    \node[place,fill=blue] (bchi) at (2.25,.4) {$11'$};
    \node[place,fill=pink] (j) at (4.5,-.2) {$3'$};

\draw (def)--(g);

\draw (j)--(a);

\draw[ultra thick,red,style=dashed] (a)--(bchi)--(def)--(a);

\draw[ultra thick,red,style=dashed] (g)--(bchi)--(j)--(g);

\draw[ultra thick] (-.5,-.2-.4)--(4.5+.5,-.2-.4)--(4.5+.5,1+.6)--(-.5,1+.6)--(-.5,-.2-.4);

\end{tikzpicture}
\caption{Component twin-width in the worst case: at worst we merge two colorful parks (with $\mathbf{cw}(G)$ vertices), and the next contraction will create a red connected component of size $2\mathbf{cw}(G)-1$.}
\end{figure}
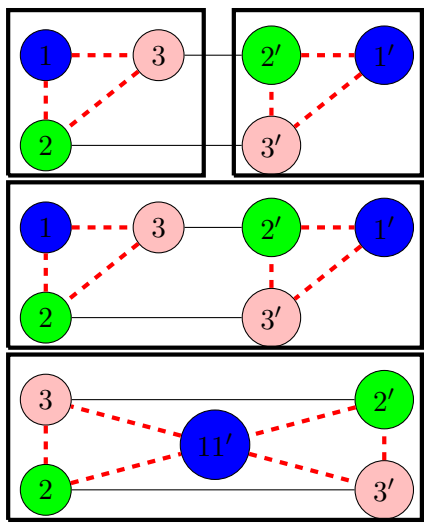

\begin{figure}[H]
\begin{tikzpicture}
\tikzstyle{place}=[draw,shape=circle];
\tikzstyle{square}=[draw,shape=rectangle];
    
    \node[place,fill=green] (a) at (0,-.2) {$2$};
    \node[place,fill=blue] (bc) at (0,1) {$1$};
    \node[place,fill=pink] (def) at (1.5,1) {$3$};
    \node[place,fill=green] (g) at (3,1) {$2'$};

\draw (def)--(g);

\draw[ultra thick,red,style=dashed] (a)--(bc)--(def)--(a);

\draw[ultra thick] (-.5,-.2-.4)--(1.5+.6,-.2-.4)--(1.5+.6,1+.6)--(-.5,1+.6)--(-.5,-.2-.4);%

\draw[ultra thick] (3-.5,1-.5)--(3+.5,1-.5)--(3+.5,1+.5)--(3-.5,1+.5)--(3-.5,1-.5);

\end{tikzpicture}

\begin{tikzpicture}
\tikzstyle{place}=[draw,shape=circle];
\tikzstyle{square}=[draw,shape=rectangle];
    
    \node[place,fill=green] (a) at (0,-.2) {$2$};
    \node[place,fill=blue] (bc) at (0,1) {$1$};
    \node[place,fill=pink] (def) at (1.5,1) {$3$};
    \node[place,fill=green] (g) at (3,1) {$2'$};

\draw (def)--(g);

\draw[ultra thick,red,style=dashed] (a)--(bc)--(def)--(a);

\draw[ultra thick] (-.5,-.2-.4)--(3+.5,-.2-.4)--(3+.5,1+.6)--(-.5,1+.6)--(-.5,-.2-.4);

\end{tikzpicture}

\begin{tikzpicture}
\tikzstyle{place}=[draw,shape=circle];
\tikzstyle{square}=[draw,shape=rectangle];
    
    \node[place,fill=green] (ag) at (0.75,-.2) {$22'$};
    \node[place,fill=blue] (bc) at (0,1) {$1$};
    \node[place,fill=pink] (def) at (1.5,1) {$3$};

\draw[ultra thick,red,style=dashed] (ag)--(bc)--(def)--(ag);

\draw[ultra thick] (-.5,-.2-.6)--(1.5+.6,-.2-.6)--(1.5+.6,1+.6)--(-.5,1+.6)--(-.5,-.2-.6);

\end{tikzpicture}
\caption{If the expression is linear, the worst case component twin-width becomes $\mathbf{lcw}(G)$.}
\end{figure}
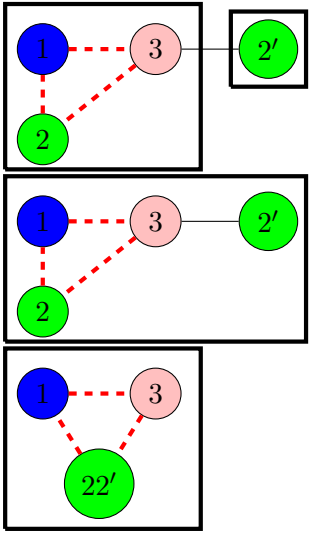

\end{multicols}

\end{document}